\documentclass[letterpaper,11pt]{article}
\usepackage[margin=1in]{geometry}
\usepackage{setspace}
\usepackage{graphicx}
\usepackage{epsfig}

\usepackage{microtype} %% Changes the look of the texts (such as spacing)
\usepackage{array} %% An extended implementation of array and tabular environments

\usepackage{amssymb}
\usepackage{amsmath}
\usepackage{amsthm}
\usepackage{mathtools} %% Some macros enhancing the texts containing a lot of mathematics
\usepackage{stmaryrd} %% Provides some additional symblos like '\mapsfrom'
\usepackage{stackrel} %% For putting things above/below math symbols
\usepackage{latexsym} %% Defines a few additional symbols like '\Box', '\Diamond'

\usepackage{tikz}
\usetikzlibrary{arrows.meta}
\usetikzlibrary{decorations.pathreplacing}
\usetikzlibrary{cd}
\usepackage[all]{xy} %% For commutative diagram, tikz.cd is recommended

\usepackage{algorithm}
\usepackage[noend]{algpseudocode}
\algnewcommand{\LineComment}[1]{\Statex\hspace{\algorithmicindent}\(\triangleright\) #1}
\algnewcommand\algorithmicforeach{\textbf{for each}}
\algdef{S}[FOR]{ForEach}[1]{\algorithmicforeach\ #1\ \algorithmicdo}
\usepackage{etoolbox} %% For the line spacing of algpseudocode

\usepackage{caption}
\usepackage{subcaption}
\usepackage[hidelinks]{hyperref} %% Handle cross-referencing to produce hypertext links 
\usepackage[symbol]{footmisc}
\usepackage{makeidx} %% Seems to be for making 'indexes' at the end of the book
\usepackage{url} %% Defines '\url'
\usepackage{color}
\usepackage{xcolor}
\usepackage{xspace} %% Defines a '\xspace' used at the end of a macro for smartly add (or not) a space char
\usepackage[normalem]{ulem} %% Various types of underlining stretchable and breakable, like '\sout', '\uline'
\usepackage{soul} %% Another package providing striking and underlining
\usepackage{ifpdf} %% Provides the \ifpdf conditional
\usepackage{extarrows}

\makeatletter
\expandafter\patchcmd\csname\string\algorithmic\endcsname{\itemsep\z@}{\itemsep=0.25ex}{}{}
\makeatother

\newcounter{usesmallsep}
\ifnum\the\value{usesmallsep}=1

    \newlength{\myitemsep}
    \newlength{\mytopsep}
    \usepackage{enumitem}
    \setlist[itemize]{leftmargin=\parindent,parsep=\parskip,
      listparindent=\parindent,itemsep=\myitemsep,topsep=\myitemsep}
    \setlist[enumerate]{leftmargin=\parindent,parsep=\parskip,
      listparindent=\parindent,itemsep=\myitemsep,,topsep=\myitemsep}
    \setlist[description]{font=\bfseries,leftmargin=\parindent,parsep=\parskip,
      listparindent=\parindent,itemsep=\myitemsep,topsep=\myitemsep}
    
    \newlength{\mypartitlesep}
    \usepackage[explicit]{titlesec}
    \titlespacing{\paragraph}{0pt}{\mypartitlesep}{\mypartitlesep}
    
    \newlength{\mythmsep}
    \newtheoremstyle{mythmstyle}
      {\mythmsep} % Space above
      {\mythmsep} % Space below
      {\itshape} % Body font
      {} % Indent amount
      {\bfseries} % Theorem head font
      {.} % Punctuation after theorem head
      {.5em} % Space after theorem head
      {} % Theorem head spec (can be left empty, meaning `normal')
    
    \newtheoremstyle{mydefstyle}
      {\mythmsep} % Space above
      {\mythmsep} % Space below
      {} % Body font
      {} % Indent amount
      {\bfseries} % Theorem head font
      {.} % Punctuation after theorem head
      {.5em} % Space after theorem head
      {} % Theorem head spec (can be left empty, meaning `normal')
    
    \theoremstyle{mythmstyle}
        \newtheorem{theorem}{Theorem}
        \newtheorem{proposition}[theorem]{Proposition}
        \newtheorem{lemma}[theorem]{Lemma}
        \newtheorem{corollary}[theorem]{Corollary}
        \newtheorem{fact}[theorem]{Fact}
        \newtheorem*{fact*}{Fact}
    \theoremstyle{mydefstyle}
        \newtheorem{definition}{Definition}
        \newtheorem{problem}{Problem}
        \newtheorem{assumption}{Assumption}
        \newtheorem{remark}{Remark}
        \newtheorem{algr}[algorithm]{Algorithm}
    
    \newenvironment{proof}
        {\vspace{-0.9em}\begin{proof}}
        {\end{proof}\vspace{-0.4em}}

\else
    \theoremstyle{plain}
        \newtheorem{theorem}{Theorem}
        \newtheorem{proposition}[theorem]{Proposition}

        \newtheorem{algr}{Algorithm}
        \newtheorem*{algr*}{Algorithm}
        \newtheorem{definition}[theorem]{Definition}
    \theoremstyle{definition}

        \newtheorem{remark}[theorem]{Remark}
        
    \usepackage{enumitem}
    \setlist[itemize]{leftmargin=\parindent}
    \setlist[enumerate]{leftmargin=\parindent}
    \setlist[description]{font=\bfseries,leftmargin=\parindent}

\fi

\newcommand{\Hm}{\mathsf{H}}
\newcommand{\Hmr}{\mathsf{\tilde{H}}}

\newcommand{\Zyc}{\mathsf{Z}}

\newcommand{\Hom}{\mathsf{Hom}}

\newcommand{\Real}{\mathbb{R}}

\newcommand{\fsimp}[2]{\sigma_{#2}}
\newcommand{\morph}[3]{\varphi^{#3}_{#2}}

\newcommand{\rank}{\mathsf{rank}\,}
\renewcommand{\ker}{\mathsf{ker}}
\newcommand{\coker}{\mathsf{coker}}
\newcommand{\img}{\mathsf{img}}

\newcommand{\cof}{\mathsf{cof}}

\newcommand{\Pers}{\mathsf{Pers}}

\newcommand{\pinds}{\mathsf{P}}
\newcommand{\ninds}{\mathsf{N}}

\newcommand{\cl}[1]{\mathrm{cls}(#1)}

\renewcommand{\bar}[1]{\overline{#1}}

\newcommand{\inv}{^{-1}}

\newcommand{\lbarrowspace}{\;}

\let\leftrightarrowsp\lrarrowsp

\let\leftarrowsp\larrowsp

\let\rightarrowsp\rarrowsp

\newcommand{\incto}{\hookrightarrow}

\newcommand{\bakincto}{\hookleftarrow}

\newcommand{\given}{\,|\,}
\newcommand{\Set}[1]{\{#1\}}
\newcommand{\bigSet}[1]{\big\{#1\big\}}

\newcommand{\bcgraph}{\mathbb{G}_\mathrm{B}}
\newcommand{\bcforest}{T}
\newcommand{\vsp}{V}

\newcommand{\pathwl}[4]{(#1\rightsquigarrow #2)_{[#3,#4]}}

\let\emptyset\varnothing
\let\intersect\cap
\let\intsec\intersect
\let\union\cup
\let\bigunion\bigcup
\newcommand{\Acal}{\mathcal{A}}
\newcommand{\Bcal}{\mathcal{B}}

\newcommand{\Ecal}{\mathcal{E}}

\newcommand{\Fcal}{\mathcal{F}}

\newcommand{\Ical}{\mathcal{I}}

\newcommand{\Mcal}{\mathcal{M}}

\newcommand{\Scal}{\mathcal{S}}

\newcommand{\Ucal}{\mathcal{U}}

\newcommand{\Xcal}{\mathcal{X}}

\newcommand{\Dbb}{\mathbb{D}}
\newcommand{\Fbb}{\mathbb{F}}

\newcommand{\Zbb}{\mathbb{Z}}

\newcommand{\aG}{\alpha}
\newcommand{\bG}{\beta}

\newcommand{\DG}{\Delta}

\newcommand{\iG}{\iota}

\newcommand{\lG}{\lambda}
\newcommand{\LG}{\Lambda}

\newcommand{\oG}{\omega}
\newcommand{\OG}{\Omega}

\newcommand{\sG}{\sigma}

\newcommand{\tG}{\tau}
\newcommand{\thG}{\theta}

\newcommand{\UG}{\Upsilon}
\newcommand{\zG}{\zeta}

\newcommand{\Dim}{p}
\newcommand{\diml}{q}
\newcommand{\birth}{b}
\newcommand{\death}{d}
\newcommand{\filtcnt}{m}

\newcommand{\simpset}{\DG}
\newcommand{\dfilt}{\Ecal}
\newcommand{\comp}[1]{D^#1}
\newcommand{\compcmplx}[1]{C^#1}
\newcommand{\compcnt}{r}
\newcommand{\rfilt}[1]{\Xcal^#1}
\newcommand{\Gres}{\Gamma}
\newcommand{\upos}{\Ucal}

\newcommand{\find}{\mathtt{find}}

\begin{document}

\title{Computing Zigzag Persistence on Graphs in Near-Linear Time\thanks{This research is supported by NSF grants CCF 1839252 and 2049010.}}

\author{Tamal K. Dey\hspace{8em}Tao Hou\vspace{1em}\\
% {\footnotesize\red
% Department of Computer Science and Engineering, The Ohio State University.
% \texttt{dey.8,hou.332@osu.edu}
% }\\
{\footnotesize 
Department of Computer Science, Purdue University. 
\texttt{tamaldey,hou145@purdue.edu}
}
}

\date{}

\maketitle
\thispagestyle{empty}

\begin{abstract}
Graphs model real-world circumstances
in many applications where
they may constantly change to capture the dynamic behavior of the phenomena.
Topological persistence
which provides a set of birth and death pairs for the topological features
is one instrument for analyzing such changing graph data.
However,
standard persistent homology defined over a growing space
cannot always capture such a dynamic process unless shrinking with
deletions is also allowed.
Hence, {\it zigzag persistence} which incorporates both insertions and
deletions of simplices is more
appropriate in such a setting. 
Unlike standard persistence which admits nearly linear-time
algorithms for graphs, 
such results for the zigzag version 
improving the general $O(m^\omega)$ time complexity
are not known, where
$\omega< 2.37286$ is the matrix multiplication exponent.
In this paper, we propose algorithms for zigzag persistence
on graphs which run in near-linear time.
Specifically, 
given a filtration with $m$ additions and deletions 
on a graph with $n$ vertices and edges,
the algorithm for $0$-dimension
runs in $O(m\log^2 n+m\log m)$ time
and the algorithm for 1-dimension
runs in $O(m\log^4 n)$ time. 
The algorithm for $0$-dimension draws upon another
algorithm designed originally for pairing critical points of Morse functions on $2$-manifolds.
The algorithm for $1$-dimension pairs a negative edge
with the {\it earliest} positive edge so that 
a $1$-cycle containing both edges 
resides in all intermediate graphs.
Both algorithms achieve the claimed time complexity 
via dynamic graph data structures 
proposed by Holm et al.
In the end, using Alexander duality, we extend the algorithm for $0$-dimension
to compute the $(p-1)$-dimensional zigzag persistence for $\mathbb{R}^p$-embedded
complexes in $O(m\log^2 n+m\log m+n\log n)$ time.
\end{abstract}

\newpage
\setcounter{page}{1}

\section{Introduction}

% Graphs are common abstraction of real-world phenomena
% and appear in many applications,
Graphs appear in many applications 
as abstraction of real-world phenomena,
where vertices represent certain objects
and edges represent their relations. 
Rather than being stationary, 
graph data obtained in applications
% \sout{that model the real-world applications}
usually change with respect to
some parameter such as time. 
A summary of these changes in a quantifiable
manner can help gain insight into the data.
% analyze the underlying phenomena. 
Persistent homology~\cite{carlsson2010zigzag,edelsbrunner2000topological}
is a suitable tool for this goal because it quantifies 
% when topological features get born and die 
the life span of topological features
as the graph changes.
%Two major tools invented by the TDA community, i.e.,
%standard (non-zigzag)~\cite{edelsbrunner2000topological} 
%and zigzag~\cite{carlsson2010zigzag} persistent homology,
%are suitable for capturing topological features 
%that get born and die in a sequence of graphs.
One drawback of using standard non-zigzag persistence~\cite{edelsbrunner2000topological} 
is that it
only allows addition of vertices and edges during the change,
% in the sequence,
whereas deletion may also happen in practice.
% changes may happen in both ways.
For example, 
many complex systems such as social networks, food webs, 
or disease spreading
are modeled by
the so-called ``dynamic networks''~\cite{holme2012temporal,kim2017stable,skarding2020foundations},
where
vertices and edges can appear and disappear at different time.
A variant of the standard persistence called 
{\it zigzag persistence}~\cite{carlsson2010zigzag} is thus a more natural tool in such scenarios because
simplices can be both added and deleted.
Given a sequence of graphs possibly with additions and deletions
(formally called a {\it zigzag filtration}),
zigzag persistence produces a set of intervals termed as {\it zigzag barcode}
in which each interval 
% encodes the birth and death 
registers the birth and death time
of a homological feature.
% in dimension zero or one.
Figure~\ref{fig:dm_net} gives an example of a graph sequence in which
% the graphs can be separated into four prominent clusters.
clusters may split (birth of 0-dimensional features) 
or vanish/merge (death of 0-dimensional features).
Moreover, addition of edges within the clusters creates 1-dimensional cycles
and deletion of edges makes some cycles disappear.
These births and deaths 
% of these topological features 
are captured by zigzag persistence.

\begin{figure}
  \centering
  \includegraphics[width=\linewidth]{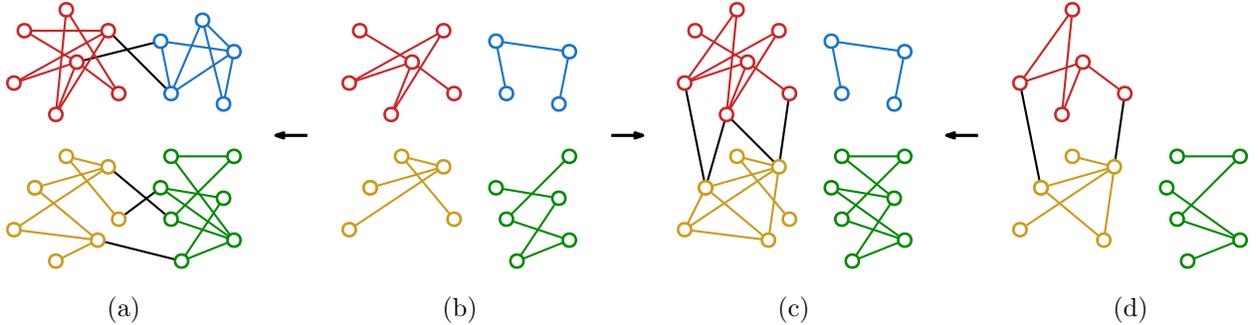}
  \caption{A sequence of graphs with four prominent clusters 
  each colored differently. Black edges connect
  different clusters and forward (resp. backward) arrows 
  indicate additions (resp. deletions) of vertices and edges. 
  From (a) to (b), two clusters
  split; from (b) to (c), two clusters merge; from (c) to (d), one cluster
  disappears. 
%   The addition (resp. deletion) also creates (resp. removes) 
%   1-dimensional cycles.
}
  \label{fig:dm_net}
\end{figure}

Algorithms for both zigzag and non-zigzag persistence
have a general-case time complexity of 
$O(m^\oG)$~\cite{carlsson2009zigzag-realvalue,edelsbrunner2000topological,maria2014zigzag,milosavljevic2011zigzag},
where
$m$ is the length of the input filtration and $\oG<2.37286$ is the matrix multiplication exponent~\cite{alman2021refined}. 
For the special
case of graph filtrations, it is well known that
non-zigzag persistence 
can be computed in $O(m\,\aG(m))$ time,
where $\aG(m)$ is the inverse Ackermann's function
%can be treated as linear for practical input
%because it 
that is almost constant for all practical purposes~\cite{CLRS3rd}. 
However, analogous faster algorithms for
% result improving
% the general case time bound for 
zigzag persistence on graphs are not known.
%to admit faster algorithms
%other than the general ones running in $O(m^\oG)$~\cite{milosavljevic2011zigzag}
%or $O(m^3)$~\cite{carlsson2009zigzag-realvalue,maria2014zigzag} time.
In this paper, 
we present algorithms for zigzag persistence on graphs with near-linear time complexity.
% we propose algorithms for zigzag persistence on graphs 
% with time complexity faster than the general case. 
In particular, given a zigzag filtration of length $m$
for a graph with $n$ vertices and edges, 
% for the general case.
%In addition to assuming the length of the given filtration to be $m$,
%we also let each complex in the filtration to be a subcomplex of 
%an {\red ambient} complex of size $n$. With these assumptions, 
our algorithm for 0-dimension
runs in $O(m\log^2 n+m\log m)$ time,
and our algorithm for $1$-dimension
runs in $O(m\log^4 n)$ time. 
Observe that the 
algorithm for $0$-dimension works for arbitrary complexes 
by restricting to the $1$-skeletons.
% For brevity, henceforth we will call $\Dim$-dimensional zigzag persistence
% as {\it $\Dim$-th} zigzag persistence instead.

The difficulty in designing faster zigzag persistence algorithms for the special
case of graphs lies in the deletion of vertices and edges.
For example, besides merging into bigger ones, 
connected components can also split into smaller ones
because of edge deletion.
Therefore,
one cannot simply kill the younger component during merging
as in standard persistence~\cite{edelsbrunner2000topological},
but rather has to pair the {\it merge} 
and {\it departure} events with the {\it split} and {\it entrance} 
events (see Sections~\ref{sec:0-zigzag} for details).
Similarly, in dimension one, deletion of edges may kill 1-cycles
so that one has to properly pair the creation and destruction
of 1-cycles, instead of simply treating all 1-dimensional intervals
as infinite ones. 

Our solutions are as follows: in dimension zero,
we find that the $O(n\log n)$ algorithm by Agarwal et al.~\cite{agarwal2006extreme} 
originally designed for pairing critical points of Morse functions on 2-manifolds
can be utilized
in our scenario. We formally prove the correctness of applying
the algorithm and use a {\it dynamic connectivity} data
structure~\cite{holm2001poly} to achieve the claimed complexity.
In dimension one, we observe that a positive and a negative edge
can be paired by finding the 
{\it earliest} 1-cycle containing both edges 
which resides in all intermediate graphs.
We further reduce the pairing to finding the {\it max edge-weight}
of a path in a minimum spanning forest.
Utilizing a data structure for {\it dynamic minimum spanning forest}~\cite{holm2001poly},
% Utilizing again the dynamic connectivity data
% structure~\cite{holm2001poly} to determine the existence
% of 1-cycles and replicating some intermediate graphs, 
we achieve the claimed time complexity.
Section~\ref{sec:1-zigzag} details this algorithm.

Using Alexander duality,
we also extend the algorithm for $0$-dimension
to compute $(\Dim-1)$-dimensional zigzag for $\Real^\Dim$-embedded
complexes. 
The connection between these two cases
% \sout{by Alexander duality}
for non-zigzag persistence
is well known~\cite{edelsbrunner2012alexander,schweinhart2015statistical},
and the challenge comes in adopting
this duality to the zigzag setting while maintaining
an efficient time budget. 
With the help of a {\it dual filtration}
% relations between 
% the dual graph and voids 
and an observation about faster void boundary reconstruction
for {\it$(\Dim-1)$-connected} complexes~\cite{dey2020computing},
we achieve a time complexity of
$O(m\log^2 n+m\log m+n\log n)$.
%efficiently constructing the dual graph
%and more specifically
%in detecting the void boundaries of the complex}. A naive
%approach for reconstructing void boundaries with nesting structures takes $O(n^2)$ time.
%With some key observations, we avoid building the nesting structure for the void boundaries
%and {\red also with an observation on duality, we avoid computing them
%repeatedly along the filtration (meaning?)}. This helps us to achieve a time complexity of
%{\red$O(m\log^2 n+m\log m+n\log n)$}.
% The time complexity of this algorithm is 
% the same as the 0-dimensional algorithm, i.e., 
%{\red$O(m\log^2 n+m\log m+n\log n)$},
% where $m$ is the length of the filtration and $n$ is the size
% of the {\red ambient} complex. 
%where $m$ and $n$ are as previously defined.
%Because the 0-th zigzag module derived by
%Alexander duality uses reduced homology, we first
%adjust the 0-dimensional algorithm to work on reduced homology.
%The most crucial step for achieving the claimed time complexity is 
%the dual graph construction because a naive
%algorithm, which depends on the nesting test of the void boundaries,
%runs in $O(n^2)$ time.
%Our key observation is that the $(\Dim-1)$-th zigzag barcode can
%be computed separately on each {\it $(\Dim-1)$-connected component}
%of the {\red ambient} complex. Because void boundaries can be
%fully reconstructed without the nesting test for $(\Dim-1)$-connected
%complexes~\cite{dey2020computing},
%our algorithm runs in the specified time.

% \subparagraph*{}
% {\red\sf All omitted proofs in this version appear
% in the full version~{\red [XX]}.}

\paragraph*{Related works.}
The algorithm for computing persistent homology by
Edelsbrunner et al.~\cite{edelsbrunner2000topological} 
is a cornerstone of topological data analysis. Several extensions
followed after this initial development.
De Silva et al.~\cite{de2011persistent} proposed to compute 
persistent {\it cohomology} instead of homology
which gives the same barcode.
De Silva et al.~\cite{de2011dualities} then showed
that the persistent cohomology algorithm runs faster in practice
than the version that uses homology.
% though having the same asymptotic complexity.
The {\it annotation} technique proposed by Dey et al.~\cite{dey2014computing}
implements the cohomology algorithm by maintaining a cohomology basis more
succinctly and extends to {\it towers} connected by simplicial maps.
% {\red (difference to the previous algorithms. complexity?)}.
These algorithms 
% mentioned so far 
run in $O(m^3)$ 
time.
% {\red Who proposed the $O(m\aG(m))$ algorithm for 0-dimension?}

Carlsson and de Silva~\cite{carlsson2010zigzag}
introduced zigzag persistence
as an extension of the standard persistence,
where they also presented a
decomposition algorithm for computing zigzag barcodes
on the level of vector spaces and linear maps. 
This algorithm is then adapted to zigzag filtrations at simplicial level by
Carlsson et al.~\cite{carlsson2009zigzag-realvalue} 
with a time complexity of $O(m^3)$.
Both algorithms~\cite{carlsson2010zigzag,carlsson2009zigzag-realvalue} 
utilize a construct called {\it right filtration} and a {\it birth-time vector}.
Maria and Oudot~\cite{maria2014zigzag} proposed 
an algorithm for zigzag persistence
based on some {\it diamond principles}
where an inverse non-zigzag filtration is always maintained
during the process.
The algorithm in~\cite{maria2014zigzag} is shown to run faster
in experiments than the algorithm in~\cite{carlsson2009zigzag-realvalue} 
though the time complexities remain the same.
Milosavljevi\'{c} et al.~\cite{milosavljevic2011zigzag}
proposed an algorithm for zigzag persistence based on
matrix multiplication
which runs in $O(m^\oG)$ time, giving the best asymptotic bound 
% currently known
for computing zigzag and non-zigzag persistence
in general dimensions.

The algorithms reviewed so far 
are all for general dimensions 
and many of them
are based on matrix operations.
Thus, it is not surprising that
the best time bound achieved is $O(m^\oG)$
% for general dimensions/
given that
computing Betti numbers 
for a simplicial $2$-complex of size $m$
is as hard as computing the rank of 
a $\Zbb_2$-matrix with $m$ non-zero entries
as shown by Edelsbrunner and Parsa~\cite{edelsbrunner2014computational}.
To lower the complexity,
one strategy 
(which is adopted by this paper)
is to consider special cases
where matrix operations can be avoided.
The work by
Dey~\cite{dey2019computing} is probably most related to ours
in that regard,
who proposed an $O(m\log m)$ algorithm for non-zigzag persistence
induced from height functions on $\Real^3$-embedded complexes.
% {\red To achieve the complexity, Dey~\cite{dey2019computing} drew upon several duality results
% known for persistence and efficiently swept the levelsets of the height function 
% by exploiting the embedding assumption.}

\section{Preliminaries}

% \paragraph*{\sout{Homology}}
% \sout{\red Use $\Hm_\Dim$ to denote both the simplicial and singular homology.
% Reduced homology.
% View graph as a 1-complex, 0-simplex vertex, 1-simplex edge.}

% \paragraph*{\sout{Zigzag persistence}}

A {\it zigzag module} (or {\it module} for short)
is a sequence of vector spaces 
\[\Mcal: \vsp_0 \leftrightarrowsp{\psi_0} \vsp_1 \leftrightarrowsp{\psi_1} 
\cdots \leftrightarrowsp{\psi_{\filtcnt-1}} \vsp_\filtcnt\]
in which
each $\psi_i$ is either a forward linear map $\psi_i:\vsp_i\to\vsp_{i+1}$
or a backward linear map $\psi_i:\vsp_i\leftarrow\vsp_{i+1}$.
We assume 
vector spaces are over field $\Zbb_2$ in this paper.
% {\red\sout{``$\aG_i\leftrightarrow\aG_{i+1}$ by $\psi_i$''.}}
% {\red For $\aG_i\in\vsp_i$ and $\aG_{i+1}\in\vsp_{i+1}$,
% we use ``$\aG_i\leftrightarrow\aG_{i+1}$ by $\psi_i$''
% to denote either $\aG_i\mapsto\aG_{i+1}$ or $\aG_i\mapsfrom\aG_{i+1}$
% based on the direction of $\psi_i$.}
A module $\Scal$ of the form
\[\Scal:W_0 \leftrightarrowsp{\phi_0} W_1 \leftrightarrowsp{\phi_1} 
\cdots \leftrightarrowsp{\phi_{\filtcnt-1}} W_\filtcnt\]
is called a {\it submodule} of $\Mcal$ if each $W_i$ is a subspace of $\vsp_i$ and
each $\phi_i$
% {\red which has the same direction as $\psi_i$}
is the restriction of $\psi_i$.
For an interval $[b,d]\subseteq[0,\filtcnt]$,
$\Scal$ is called an {\it interval submodule} of $\Mcal$ over $[b,d]$
% an {\it interval submodule} $\Ical^{[b,d]}$ of $\Mcal$,
% $\Ical^{[b,d]}:W_0 \leftrightarrowsp{\phi_0} 
% \cdots \leftrightarrowsp{\phi_{\filtcnt-1}} W_\filtcnt$,
% is a submodule such that 
if $W_i$
is one-dimensional for $i\in[\birth,\death]$
and is trivial for $i\not\in[\birth,\death]$,
and $\phi_i$ is an isomorphism for $i\in[\birth,\death-1]$.
It is well known~\cite{carlsson2010zigzag} that $\Mcal$ admits 
an {\it interval decomposition}
$\Mcal=\bigoplus_{\aG\in\Acal}\Ical^{[\birth_\aG,\death_\aG]}$
which is a direct sum of interval submodules of $\Mcal$.
The (multi-)set of intervals
$\Set{[\birth_\aG,\death_\aG]\given \aG\in\Acal}$
is called the {\it zigzag barcode} (or {\it barcode} for short) of $\Mcal$
and is denoted as $\Pers(\Mcal)$.
Each interval in a zigzag barcode is called a {\it persistence interval}.

% For computation purposes,
In this paper,
we mainly focus on a special type of zigzag modules:

\begin{definition}[Elementary zigzag module]
A zigzag module is called \textbf{elementary} if 
it starts with the trivial vector space
and all linear maps in the module
are of the three forms: {\rm(}{\sf i}{\rm)} an isomorphism; 
{\rm(}{\sf ii}{\rm)} an injection with rank 1 cokernel; 
{\rm(}{\sf iii}{\rm)} a surjection with rank 1 kernel.
\end{definition}

A {\it zigzag filtration} (or {\it filtration} for short)
is a sequence of simplicial complexes
\[\Fcal: K_0 \leftrightarrowsp{\fsimp{\Fcal}{0}} K_1 \leftrightarrowsp{\fsimp{\Fcal}{1}}
\cdots \leftrightarrowsp{\fsimp{\Fcal}{\filtcnt-1}} K_\filtcnt\]
in which each
$K_i\leftrightarrowsp{\fsimp{\Fcal}{i}} K_{i+1}$ is either a forward inclusion
$K_i\incto K_{i+1}$ with a single simplex $\fsimp{\Fcal}{i}$ added,
or a backward inclusion $K_i\bakincto K_{i+1}$
with a single $\fsimp{\Fcal}{i}$ deleted.
% and $K_i,K_{i+1}$ differ by only one simplex denoted $\fsimp{\Fcal}{i}$.
When the $\fsimp{\Fcal}{i}$'s
% added or deleted 
are not explicitly used,
we drop them and simply denote $\Fcal$ as
$\Fcal: K_0 \leftrightarrow K_1 \leftrightarrow
\cdots \leftrightarrow K_\filtcnt$.
For computational purposes,
we sometimes assume that a filtration starts with the empty complex,
i.e., $K_0=\emptyset$ in $\Fcal$.
Throughout the paper, we also assume that each $K_i$ in $\Fcal$ is a subcomplex 
of a fixed complex $K$; such a $K$, when not given, 
can be constructed by taking the union
of every $K_i$ in $\Fcal$.
In this case, we call $\Fcal$ a filtration {\it of} $K$.

Applying the $\Dim$-th homology with $\Zbb_2$ coefficients on $\Fcal$, 
we derive the
{\it $\Dim$-th zigzag module of $\Fcal$}
\[\Hm_\Dim(\Fcal): 
\Hm_\Dim(K_0) \leftrightarrowsp{\morph{\Fcal}{0}{\Dim}} 
\Hm_\Dim(K_1) \leftrightarrowsp{\morph{\Fcal}{1}{\Dim}} 
\cdots 
% \leftrightarrowsp{\morph{\Fcal}{\filtcnt-2}} \Hm_\Dim(K_{\filtcnt-1}) 
\leftrightarrowsp{\morph{\Fcal}{\filtcnt-1}{\Dim}} \Hm_\Dim(K_\filtcnt)\]
in which each $\morph{\Fcal}{i}{\Dim}$ is the linear map
induced by the inclusion.
In this paper, 
whenever $\Fcal$ is used to denote a filtration,
we use $\morph{\Fcal}{i}{\Dim}$ to denote a linear map
in the module $\Hm_\Dim(\Fcal)$.
Note that $\Hm_\Dim(\Fcal)$ is an elementary module 
if $\Fcal$ starts with an empty complex.
Specifically, we call $\Pers(\Hm_\Dim(\Fcal))$ 
the {\it $\Dim$-th zigzag barcode} of $\Fcal$.

\section{Zero-dimensional zigzag persistence}
\label{sec:0-zigzag}

% \subsection{Algorithm}

We present our algorithm for 0-th zigzag 
persistence\footnote{For brevity, 
henceforth we call $\Dim$-dimensional zigzag persistence
as {\it $\Dim$-th} zigzag persistence.}
% (see Algorithm~\ref{alg:0-zigzag}) 
in this section.
The input is assumed to be on graphs
but note that our algorithm can be applied to any complex by restricting to its 1-skeleton.
We first define the barcode graph of
a zigzag filtration
which is a construct that our algorithm implicitly works on.
In a barcode graph,
nodes correspond to 
connected components of graphs in the filtration
and edges encode the mapping between the components:

\begin{figure} 
    \begin{subfigure}[t]{0.9\textwidth}
    % \centering
    \includegraphics[width=\textwidth]{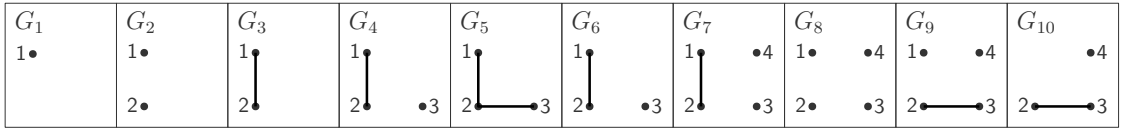} 
    \vspace{-1em}
    \caption{A zigzag filtration of graphs 
    with 0-th barcode $\Set{[2,2],[4,4],[6,8],[8,9],[7,10],[1,10]}$.}
    \label{fig:filt}
    \end{subfigure}

    % \hfill 
    \vspace{1.5em}

    \begin{subfigure}[t]{0.65\textwidth}
    \centering
    \includegraphics[width=\textwidth]{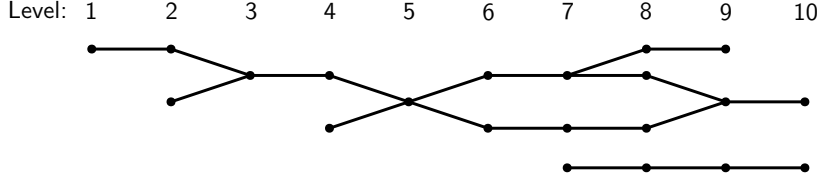} 
    \vspace{-1em}
    \caption{The barcode graph for the filtration shown in Figure~\ref{fig:filt}.}
    \label{fig:bcgraph}
    \end{subfigure}

    \vspace{1em}

    \begin{subfigure}[t]{\textwidth}
    \centering
    \includegraphics[width=\textwidth]{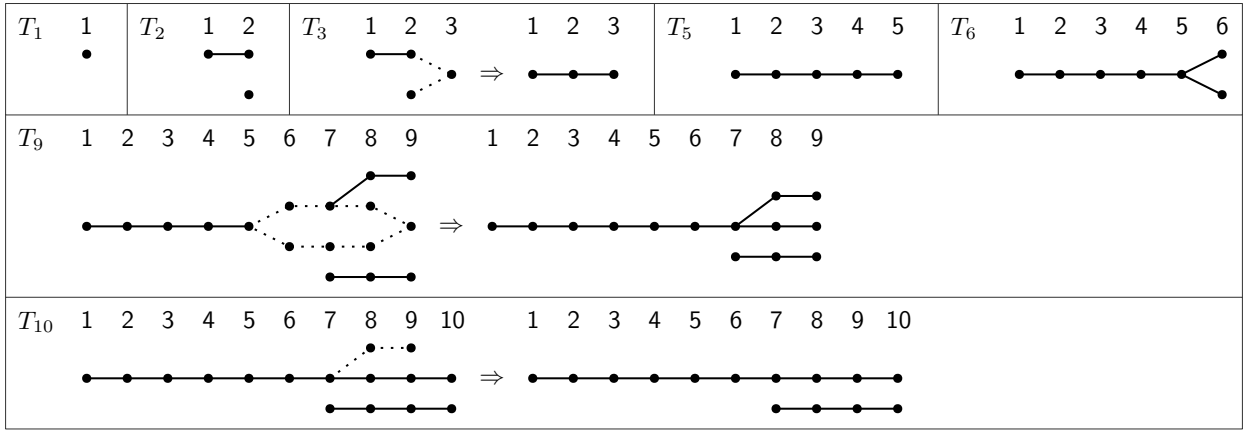} 
    \vspace{-1em}
    \caption{Barcode forests constructed in Algorithm~\ref{alg:0-zigzag} 
    for the barcode graph in Figure~\ref{fig:bcgraph}. 
    For brevity, some forests are skipped.
    The horizontally arranged labels indicate the levels.}
    \label{fig:bcforest}
    \end{subfigure}

\caption{Examples of a zigzag filtration, a barcode graph, and barcode forests.}
\label{fig:example} 
\end{figure}

\begin{definition}[Barcode graph]\label{dfn:bcgraph}
For a graph $G$ and a zigzag filtration
$\Fcal: G_0 \leftrightarrow G_1 \leftrightarrow 
\cdots \leftrightarrow G_\filtcnt$ of $G$,
the 
% {\red(0-th)} 
\textbf{barcode graph} $\bcgraph(\Fcal)$ of $\Fcal$
is a graph whose vertices {\rm(}preferably called \textbf{nodes}{\rm)}
are associated with a \textbf{level}
and whose edges connect nodes only at adjacent levels.
The graph $\bcgraph(\Fcal)$ is constructively described as follows:
\begin{itemize}
    \item For each $G_i$ in $\Fcal$ and each connected component of $G_i$,
    there is a node in $\bcgraph(\Fcal)$ at {level}~$i$ corresponding
    to this component; 
    this node is also called a \textbf{level-$i$ node}.
    % To differentiate from vertices in a simplicial complex,
    % we call a vertex in a barcode graph as a \textbf{node}.
    % A node at level $i$ is also called a level-$i$ node.
    \item For each inclusion $G_i\leftrightarrow G_{i+1}$ in $\Fcal$,
    if it is forward, 
    then there is an edge connecting a level-$i$ node $v_i$ 
    to a level-$(i+1)$ node $v_{i+1}$ if and only if 
    the component of $v_i$ maps to the component of $v_{i+1}$ by the inclusion.
    Similarly, if the inclusion is backward, then 
    $v_i$ connects to $v_{i+1}$ by an edge 
    % \sout{if and only if} 
    iff
    the component of $v_{i+1}$ maps to the component of $v_{i}$.
\end{itemize}
For two nodes at different levels in $\bcgraph(\Fcal)$,
% or $\dbcgraph(\Fcal)$, 
the node at the higher {\rm(}resp. lower{\rm)} level 
is said to be \textbf{higher} {\rm(}resp. \textbf{lower}{\rm)} than the other.
% Note that whenever we mention a barcode graph without specifying whether
% it is directed or undirected, we always mean an undirected one.
% The lowest node of a connected component in $\bcgraph(\Fcal)$ 
% is called the {\it root} of the component;
% The highest nodes of $\bcgraph(\Fcal)$ are called \textbf{leaves}.
\end{definition}
\begin{remark}
Note that some works~\cite{dey2019computing,kim2017stable} also have used similar notions of barcode graphs.
\end{remark}
% \begin{remark}
% For a zigzag filtration $\Fcal$, the barcode graph $\bcgraph(\Fcal^i)$
% of the {\red prefix} $\Fcal^i$ is the full subgraph of $\bcgraph(\Fcal)$
% induced by nodes whose levels are between 1 and $i$.
% \end{remark}

Figure~\ref{fig:filt} and~\ref{fig:bcgraph} give an example of 
a zigzag filtration and its barcode graph.
% where one can verify that each level-$i$ node corresponds to 
% a connected component of $G_i$ and the edges 
% encode the mapping of connected components by the inclusion.
Note that a barcode graph is of size $O(mn)$,
where $m$ is the length of $\Fcal$ 
and $n$ is the number of vertices and edges
of $G$. Although we present our algorithm
(Algorithm~\ref{alg:0-zigzag}) by first building the barcode graph,
the implementation does not do so explicitly, allowing us
to achieve the claimed time complexity; 
see Section~\ref{sec:0-zz-imp}
% the end of this subsection 
for the implementation details.
Introducing barcode graphs
helps us justify the algorithm,
and more importantly, points to the fact 
that the algorithm can be applied whenever such a barcode graph can be built.

\begin{algr}[Algorithm for 0-th zigzag persistence]
\label{alg:0-zigzag}
\begin{itemize}\item[]\end{itemize}\noindent
Given a graph $G$ and a zigzag filtration
$\Fcal: \emptyset=G_0 \leftrightarrow G_1 \leftrightarrow 
\cdots \leftrightarrow G_\filtcnt$ of $G$,
we first build the barcode graph $\bcgraph(\Fcal)$,
% of $\Fcal$,
and then apply the pairing algorithm described in~\cite{agarwal2006extreme}
on $\bcgraph(\Fcal)$ to compute $\Pers(\Hm_0(\Fcal))$.
For a better understanding, 
we rephrase this algorithm %~\cite{agarwal2006extreme}
which originally works on Reeb graphs:

The algorithm iterates for $i=0,\ldots,\filtcnt-1$
and maintains a \textbf{barcode forest} $\bcforest_i$,
whose leaves have a one-to-one correspondence 
to level-$i$ nodes of $\bcgraph(\Fcal)$.
Like the barcode graph,
each tree node in a barcode forest is associated with a level and each tree edge
connects nodes at adjacent levels.
For each tree in a barcode forest,
the lowest node is the root.
% and the highest nodes are the leaves.
% {\red A node with more than one child in a barcode forest
% is called a {\it splitting} node.}
Initially, $\bcforest_0$ is empty;
then, the algorithm builds $\bcforest_{i+1}$ from $\bcforest_{i}$
in the $i$-th iteration.
Intervals for $\Pers(\Hm_0(\Fcal))$ are produced while updating the barcode forest.
{\rm(}Figure~\ref{fig:bcforest} illustrates such updates.{\rm)}

Specifically, the $i$-th iteration proceeds as follows:
% processes the connections of level-$i$ and{}
% level-$(i+1)$ nodes in $\bcgraph(\Fcal)$.
% The forest
first, $\bcforest_{i+1}$ is formed by copying 
the level-$(i+1)$ nodes of $\bcgraph(\Fcal)$ and their connections
to the level-$i$ nodes, into $\bcforest_{i}$;
the copying is possible because leaves of $T_i$ and 
level-$i$ nodes of $\bcgraph(\Fcal)$ have a one-to-one correspondence;
see transitions from $T_5$ to $T_6$ and from $T_9$ to $T_{10}$
in Figure~\ref{fig:bcforest}.
% Note that this tentative $\bcforest_{i+1}$ may now contain loops
% and 
We further change $\bcforest_{i+1}$ 
% {\red into a tree} 
under the following events:
\vspace{0.3em}\begin{description}
    \item[Entrance:] One level-$(i+1)$ node in $\bcforest_{i+1}$, said to be \textbf{entering},
    % which is 
    does not connect to any level-$i$ node.
    
    \item[Split:] One level-$i$ node in $\bcforest_{i+1}$,
    % which is 
    said to be \textbf{splitting},
    connects to
    two different level-$(i+1)$ nodes.
    For the two events so far,
    no changes need to be made on $\bcforest_{i+1}$.
    
   \item[Departure:] One level-$i$ node $u$ in $\bcforest_{i+1}$,
    % which is 
    said to be \textbf{departing},
    does not connect to any level-$(i+1)$ node.
    If $u$ has splitting ancestors {\rm(}i.e., ancestors which are also splitting nodes{\rm)}, add an 
    % ({\red open-closed}) 
    interval $[j+1,i]$ to $\Pers(\Hm_0(\Fcal))$, where
    $j$ is the level of the highest splitting ancestor $v$ of $u$;
    otherwise, add an 
    % (closed-closed) 
    interval $[j,i]$ to $\Pers(\Hm_0(\Fcal))$, where $j$ is the level of the root $v$ of $u$.
    We then delete the path from $v$ to $u$ in $\bcforest_{i+1}$.
    % \tamal{the term splitting ancestor appears
    % here for the first time and perhaps a definition in the footnote will be good.}
    
    \item[Merge:] Two different level-$i$ nodes $u_1,u_2$ in $\bcforest_{i+1}$ 
    connect to the same level-$(i+1)$ node.
    Tentatively, $\bcforest_{i+1}$ may now contain a loop
    and is not a tree.
    % $w$ and $w$ is said to be {\it merging}.
    If $u_1,u_2$ are in different trees in $\bcforest_i$, 
    add an 
    % (closed-open) 
    interval $[j,i]$ to $\Pers(\Hm_0(\Fcal))$,
    where $j$ is the level of the higher root of $u_1,u_2$ in $\bcforest_i$;
    otherwise, add an 
    % (open-open) 
    interval $[j+1,i]$ to $\Pers(\Hm_0(\Fcal))$, where $j$ is the level of 
    the highest common ancestor of $u_1,u_2$ in $\bcforest_i$.
    We then glue the two paths from $u_1$ and $u_2$ to their level-$j$ ancestors
    in $\bcforest_{i+1}$,
    after which $\bcforest_{i+1}$ is guaranteed to be a tree.
    
    \item[No-change:] If none of the above events happen, 
    no changes are made on $\bcforest_{i+1}$.
\end{description}

\vspace{-0.5em}
At the end, for each root in $\bcforest_{\filtcnt}$ at a level $j$,
add an 
% (closed-closed) 
interval $[j,\filtcnt]$ to $\Pers(\Hm_0(\Fcal))$,
and for each splitting node
in $\bcforest_{\filtcnt}$ at a level $j$,
add an 
% (open-closed) 
interval $[j+1,\filtcnt]$ to $\Pers(\Hm_0(\Fcal))$.
\end{algr}

\begin{remark}
The justification of Algorithm~\ref{alg:0-zigzag} is given in Section~\ref{sec:0-zigzag-proof}.
\end{remark}

Figure~\ref{fig:bcforest} gives examples of barcode forests
constructed by Algorithm~\ref{alg:0-zigzag} 
for the barcode graph shown in Figure~\ref{fig:bcgraph},
where $T_1$ and $T_2$ introduce entering nodes,
$T_6$ introduces a splitting node,
and $T_{10}$ introduces a departing node.
In $T_{10}$, the departure event happens and the dotted path is deleted,
producing an interval $[8,9]$.
In $T_3$ and $T_9$, the merge event happens and the dotted paths 
are glued together, producing intervals $[2,2]$ and $[6,8]$.
Note that the glued level-$i$ nodes are in different trees 
% (of $T_2$) 
in $T_3$ and are in the same tree in $T_9$.

\subsection{Implementation}
\label{sec:0-zz-imp}

As mentioned, to achieve the claimed time complexity,
we do not explicitly build the barcode graph.
Instead, we differentiate the different events as follows:
inserting (resp. deleting) a vertex in $\Fcal$ simply corresponds to
the entrance (resp. departure) event,
whereas inserting (resp. deleting) an edge corresponds to
the merge (resp. split) event only when connected components in the graph
merge (resp. split).

To keep track of the connectivity of vertices,
we use a {\it dynamic connectivity} data structure by Holm~et~al.~\cite{holm2001poly},
which we denote as $\Dbb$.
Assuming that 
$\filtcnt$ is the length of $\Fcal$ and 
$n$ is the number of vertices and edges of $G$,
the data structure $\Dbb$ supports the following operations:
\begin{itemize}
    \item Return the identifier\footnote{Since $\Dbb$ 
    maintains the connectivity information by
    dynamically updating the spanning forest for the current graph,
    the identifier of a connected component is indeed the identifier
    of a tree in the spanning forest.}
    of the connected component of a vertex $v$ in $O(\log n)$ time.
    % which can be used to determine whether two vertices are connected. 
    We denote this subroutine as $\find(v)$.
    \item Insert or delete an edge,
    and possibly update the connectivity information, in $O(\log^2 n)$ amortized time.
\end{itemize}

We also note the following implementation details:
\begin{itemize}
    \item All vertices of $G$ are added to $\Dbb$ initially
and are then never deleted. 
But we make sure that
edges in $\Dbb$ always equal edges in $G_i$
as the algorithm proceeds
so that $\Dbb$ still records the connectivity of $G_i$.
    
    \item At each iteration $i$,
we update $\bcforest_i$ 
to form $\bcforest_{i+1}$ according to the changes of the connected components
from $G_i$ to $G_{i+1}$.
For this,
we maintain a key-value map $\phi$ from connected components of $\Dbb$
to leaves of the barcode forest,
and $\phi$ is initially empty.

    \item In a barcode forest $\bcforest_i$,
since the level of a leaf
always equals $i$,
we only record the level of a non-leaf node.
Note that at iteration $i$, a leaf in $\bcforest_{i}$ may
uniquely connect to a single leaf in $\bcforest_{i+1}$. 
In this case, 
we simply let the leaf in $\bcforest_{i}$ automatically become
a leaf in $\bcforest_{i+1}$; see Figure~\ref{fig:T3toT4}.
The size of a barcode forest is then $O(m)$.
% and the traversal and updates of the barcode forest 
\end{itemize}

\begin{figure}
  \centering
  \includegraphics[width=0.33\linewidth]{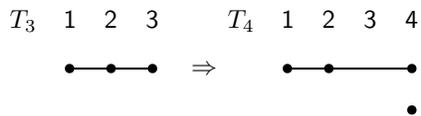}
  \caption{For the example in Figure~\ref{fig:example},
  to form $T_4$, 
  our implementation only adds a level-4 entering node,
  whereas the leaf in $T_3$ is not touched. 
  Since the level of a leaf always equals the index of the barcode forest,
  the leaf at level 3 in $T_3$ 
  automatically becomes a leaf at level 4 in $T_4$.}
  \label{fig:T3toT4}
\end{figure}

Now we can present the full detail of the implementation.
Specifically,
for each addition and deletion in $\Fcal$, 
we do the following in each case:

\begin{description}
    \item[Adding vertex $\fsimp{\Fcal}{i}=v$:]
    Add an isolated node to the barcode forest and let
    $\phi(\find(v))$ equal this newly added node.
    
    \item[Deleting vertex $\fsimp{\Fcal}{i}=v$:]
    Let $\ell=\phi(\find(v))$; then, $\ell$ is the node in the barcode 
    forest that is departing. 
    Update the barcode forest as described in Algorithm~\ref{alg:0-zigzag}.
    
    \item[Adding edge $\fsimp{\Fcal}{i}=(u,v)$:]
    Let $t_1=\find(u)$, $t_2=\find(v)$,
    $\ell_1=\phi(t_1)$, and $\ell_2=\phi(t_2)$.
    If $t_1=t_2$, then the no-change event happens;
    otherwise, the merge event happens.
    We then add $(u,v)$ to $\Dbb$.
    For the no-change event,
    do nothing after this.
    For the merge event, do the following:
    glue the paths from $\ell_1$ and $\ell_2$ to their ancestors
    as described in Algorithm~\ref{alg:0-zigzag};
    attach a new child $\ell$ to the highest
    glued node; update $\phi(\find(u))$ to be $\ell$.
    % because $\find(u)$ may have changed after adding $(u,v)$ to $\Dbb$.
    
    \item[Deleting edge $\fsimp{\Fcal}{i}=(u,v)$:]
    Let $\ell=\phi(\find(u))$, and then delete $(u,v)$ from $\Dbb$.
    If $\find(u)=\find(v)$ after this, then the no-change event happens
    but we have to update $\phi(\find(u))$ to be $\ell$ 
    because the identifiers of the connected components in $\Dbb$ may change
    after deleting the edge~\cite{holm2001poly}.
    % If $\find(u)\neq\find(v)$ after deleting $(u,v)$ from $\Dbb$,
    Otherwise, the split event happens: we attach two new
    children $\ell_1$, $\ell_2$ to $\ell$ in the barcode forest and set
    $\phi(\find(u))=\ell_1$, $\phi(\find(v))=\ell_2$.
\end{description}

\paragraph*{Mergeable trees.}
Following the idea in~\cite{agarwal2006extreme},
the barcode forest can be implemented
using the {\it mergeable trees} data structure by Georgiadis et al.~\cite{georgiadis2011data}.
Since the maximum number of nodes in a barcode forest is $O(m)$,
the data structure supports the following operations, 
each of which takes $O(\log m)$ amortized time:
\begin{itemize}
    \item Return the root of a node.
    \item Return the nearest common ancestor of two leaves (in the same tree).
    \item Glue the paths from two leaves (in the same tree) to their nearest common ancestor.
\end{itemize}

Note that 
while we delete the path from the departing node to its ancestor
in the departure event,
deletions are not supported by mergeable trees.
However,
path deletions are indeed unnecessary 
which are only meant for a clear exposition.
Hence,
during implementation, we only traverse each ancestor
of the departing node until an {\it unpaired}\footnote{An entering 
or splitting node is initially {\it unpaired} when introduced and becomes {\it paired} 
when its level is used to produce an interval. 
E.g., the node $v$ becomes paired in the departure event in Algorithm~\ref{alg:0-zigzag}.}
one is found
without actual deletions.
Since each node can only 
be traversed once, the traversal in the departure events
takes $O(m)$ time in total.
See \cite[Section 5]{georgiadis2011data}
for details of implementing the barcode forest and its operations 
using mergeable trees.

\paragraph*{Complexity.}
The time complexity of the algorithm is $O(m\log^2 n+m\log m)$
dominated by the operations of the dynamic connectivity 
and the mergeable trees data structures.

\subsection{Justification}
\label{sec:0-zigzag-proof}
In this subsection, we justify the correctness of Algorithm~\ref{alg:0-zigzag}.
For each entering node $u$ in a $\bcforest_i$ of Algorithm~\ref{alg:0-zigzag},
there must be a single
% \tao{may need to talk} 
node in $\bcgraph(\Fcal)$ at the level of $u$
with the same property. 
So we also have entering nodes in $\bcgraph(\Fcal)$.
Splitting and departing nodes
% \sout{and \red merging nodes} 
in $\bcgraph(\Fcal)$
can be similarly defined.
% All other nodes in a barcode graph or forest are called {\it regular}.

We first prepare some standard notions and facts 
in zigzag persistence
(Definition~\ref{dfn:rep-cls} and~\ref{dfn:pos-neg-inds}, 
Proposition~\ref{prop:pn-paring-w-rep})
that help with our proofs.
Some notions also appear in previous works in different forms; 
see, e.g., \cite{maria2014zigzag}.

\begin{definition}[Representatives]
\label{dfn:rep-cls}
Let 
% $\Dim\geq 0$,
$\Mcal: \vsp_0 \leftrightarrowsp{\psi_0} \cdots \leftrightarrowsp{\psi_{\filtcnt-1}} \vsp_\filtcnt$
be an elementary zigzag module
and $[s,t]\subseteq[1,\filtcnt]$ be an interval.
An indexed set
$\Set{\aG_i\in \vsp_i\given i\in[s,t]}$ 
is called a set of \textbf{partial representatives} 
% in $\Hm_\Dim(\Fcal)$ 
for $[s,t]$ if for every $i\in[s,t-1]$, 
$\aG_i\mapsto\aG_{i+1}$ or $\aG_i\mapsfrom\aG_{i+1}$ by $\psi_i$;
% {\red $\aG_i\leftrightarrow\aG_{i+1}$ by $\psi_i$};
it is called a set of \textbf{representatives} for $[s,t]$
if the following additional conditions are satisfied:
\begin{enumerate}
    \item \label{itm:dfn-rep-cls-birth}
    If $\psi_{s-1}:V_{s-1}\rightarrow V_s$ is forward
    with non-trivial cokernel,
    then $\aG_s$ is not in $\img(\psi_{s-1})$;
    if $\psi_{s-1}: V_{s-1}\leftarrow V_{s}$ is backward
    with non-trivial kernel,
    then $\aG_s$ is the non-zero element in $\ker(\psi_{s-1})$.
    
    \item \label{itm:dfn-rep-cls-death}
    If $t<\filtcnt$ and
    $\psi_{t}: V_{t}\leftarrow V_{t+1}$ is backward
    with non-trivial cokernel,
    then $\aG_t$ is not in $\img(\psi_{t})$;
    if $t<\filtcnt$
    and $\psi_{t}: V_t\rightarrow V_{t+1}$ is forward
    with non-trivial kernel,
    then $\aG_t$ is the non-zero element in $\ker(\psi_{t})$.
\end{enumerate}
Specifically, when $\Mcal:=\Hm_\Dim(\Fcal)$ for a zigzag filtration $\Fcal$,
%a set of representatives (resp. partial representatives) for an interval
%is also called a set of 
we use terms 
\textbf{$\Dim$-representatives} 
and \textbf{partial $\Dim$-representatives} to emphasize the dimension $p$.
\end{definition}
\begin{remark}
% See Figure~\ref{fig:reps} for an example of a set of 0-representatives.
Let $\Fcal$ be the filtration given in Figure~\ref{fig:filt},
and let $\aG_8$, $\aG_9$ be the sum of the component containing
vertex 1 and the component containing vertex 2 in $G_8$ and $G_9$.
Then, $\Set{\aG_8,\aG_9}$ is a set of 0-representatives 
for the interval $[8,9]\in\Pers(\Hm_0(\Fcal))$.
\end{remark}

% \begin{figure}
%   % \centering
%   \includegraphics[width=0.37\linewidth]{fig/reps}
%   \caption{A set of 0-representatives $\Set{\aG_7,\aG_8}$ 
%   for the interval $[7,8]\in\Pers(\Hm_0(\Fcal))$,
%   where $\Fcal$ is the filtration given in Figure~\ref{fig:filt}.}
%   \label{fig:reps}
% \end{figure}

\begin{definition}[Positive/negative indices]
\label{dfn:pos-neg-inds}
Let 
% $\Dim\geq 0$,
$\Mcal: \vsp_0 \leftrightarrowsp{\psi_0} \cdots \leftrightarrowsp{\psi_{\filtcnt-1}} \vsp_\filtcnt$
be an elementary zigzag module.
The set of \textbf{positive indices} of $\Mcal$,
denoted $\pinds(\Mcal)$,
and the set of \textbf{negative indices} of $\Mcal$,
denoted $\ninds(\Mcal)$,
are constructed as follows:
for each forward $\psi_i:\vsp_i\to\vsp_{i+1}$,
if $\psi_i$ is an injection with non-trivial cokernel,
add $i+1$ to $\pinds(\Mcal)$;
if $\psi_i$ is a surjection with non-trivial kernel,
add $i$ to $\ninds(\Mcal)$.
Furthermore,
for each backward $\psi_i:\vsp_i\leftarrow\vsp_{i+1}$,
if $\psi_i$ is an injection with non-trivial cokernel,
add $i$ to $\ninds(\Mcal)$;
if $\psi_i$ is a surjection with non-trivial kernel,
add $i+1$ to $\pinds(\Mcal)$.
Finally, add $\rank \vsp_\filtcnt$ copies of $m$ to $\ninds(\Mcal)$.
% The set of \textbf{positive indices} of $\Mcal$, denoted $\pinds(\Mcal)$,
% contains each $i\in[1,\filtcnt]$ such that $\psi_{i-1}$
% is a forward map and has non-trivial cokernel,
% or is a backward map and has non-trivial kernel.
% The set of \textbf{negative indices} of $\Mcal$, denoted $\ninds(\Mcal)$,
% contains each $i\in[0,\filtcnt-1]$ such that $\psi_{i}$
% is a forward map and has non-trivial kernel,
% or is a backward map and has non-trivial cokernel.
% Additionally, $\ninds(\Mcal)$ also contains $\rank \vsp_\filtcnt$ copies of~$m$.
% Specifically, for a zigzag filtration $\Fcal$ and its induced module $\Hm_\Dim(\Fcal)$,
% $\pinds(\Hm_\Dim(\Fcal))$ is called 
% the set of \textbf{$\Dim$-th positive indices} of $\Fcal$
% and is denoted by $\pinds_\Dim(\Fcal)$.
% The set $\ninds_\Dim(\Fcal)$ of \textbf{$\Dim$-th negative indices} of $\Fcal$
% is defined similarly.
\end{definition}

% \begin{remark}
% \sout{Note that $\ninds(\Mcal)$ is in fact a multi-set, though, calling it a set 
% should not cause any confusion in this paper.
% Also note that $|\pinds(\Mcal)|=|\ninds(\Mcal)|$.}
% \end{remark}

\begin{remark}
For each $\psi_i:\vsp_i\leftrightarrow\vsp_{i+1}$ in Definition~\ref{dfn:pos-neg-inds},
if $i+1\in\pinds(\Mcal)$, then $i\not\in\ninds(\Mcal)$;
similarly, if $i\in\ninds(\Mcal)$, then $i+1\not\in\pinds(\Mcal)$.
Furthermore, if $\psi_i$ is an isomorphism, 
then $i\not\in\ninds(\Mcal)$ and $i+1\not\in\pinds(\Mcal)$.
\end{remark}

Note that $\ninds(\Mcal)$ in Definition~\ref{dfn:pos-neg-inds} 
is in fact a multi-set; calling it a set 
should not cause any confusion in this paper though.
Also note that $|\pinds(\Mcal)|=|\ninds(\Mcal)|$,
and every index in $\pinds(\Mcal)$ (resp. $\ninds(\Mcal)$) 
is the start (resp. end) of an interval in $\Pers(\Mcal)$.
% Each index in $\pinds(\Mcal)$ is the start of an interval in $\Pers(\Mcal)$ and
% each index in $\ninds(\Mcal)$ is the end of an interval in $\Pers(\Mcal)$.
This explains why we add $\rank \vsp_\filtcnt$ copies of $m$ to $\ninds(\Mcal)$
because there are always $\rank \vsp_\filtcnt$ number of intervals ending with $m$
in $\Pers(\Mcal)$;
see the example in Figure~\ref{fig:filt} where $\rank\Hm_0(G_{10})=2$.
% \sout{For example,
% the 0-th barcode of the filtration in Figure~\ref{fig:filt} contains
% two intervals ending with $10$ and $\rank\Hm_0(G_{10})=2$.}
%\sout{\red We also note that the problem of computing zigzag barcode for an elementary 
%module $\Mcal$ can be viewed as
%finding a pairing of indices of $\pinds(\Mcal)$ and $\ninds(\Mcal)$ such that
%each pair is the start and end of a persistence interval.}

\begin{proposition}\label{prop:pn-paring-w-rep}
Let $\Mcal$ be an elementary zigzag module
and $\pi:\pinds(\Mcal)\to\ninds(\Mcal)$ be a bijection.
If every $b\in\pinds(\Mcal)$ satisfies that $b\leq\pi(b)$ and the interval $[b,\pi(b)]$
has a set of representatives,
then $\Pers(\Mcal)=\Set{[b,\pi(b)]\given b\in\pinds(\Mcal)}$.
\end{proposition}
% \begin{proof}
% See Appendix~\ref{sec:pf-prop-pn-paring-w-rep}.
% \end{proof}
\begin{proof}
For each $b\in\pinds(\Mcal)$, 
let $\bigSet{\aG^b_{j}\given j\in[b,\pi(b)]}$ be
a set of representatives for $[b,\pi(b)]$.
Then, define $\Ical^{[b,\pi(b)]}$ as an interval submodule of $\Mcal$ over $[b,\pi(b)]$
such that $\Ical^{[b,\pi(b)]}(j)$ 
% equals the 1-dimensional vector space 
is generated by $\aG^b_j$ 
if $j\in[b,\pi(b)]$ and is trivial otherwise,
where $\Ical^{[b,\pi(b)]}(j)$ denotes the $j$-th vector space in $\Ical^{[b,\pi(b)]}$.
We claim that 
$\Mcal=\bigoplus_{b\in\pinds(\Mcal)}\Ical^{[b,\pi(b)]}$,
which implies the proposition.
To prove this, 
suppose that $\Mcal$ is of the form
\[\Mcal: \vsp_0 \leftrightarrowsp{\psi_0} 
\vsp_1 \leftrightarrowsp{\psi_1} 
\cdots \leftrightarrowsp{\psi_{\filtcnt-1}} \vsp_\filtcnt\]
Then, we only need to verify that 
for every $i\in[0,\filtcnt]$, the set
$\bigSet{\aG^b_i\given b\in\pinds(\Mcal)\text{ and }[b,\pi(b)]\ni i}$ 
is a basis of $\vsp_i$.
% The verification is
% similar to the one in the proof of {\red Theorem 2 in [XX]},
% and is omitted.
We prove this by induction on $i$.
For $i=0$, since $\vsp_0=0$, 
$\bigSet{\aG^b_0\given b\in\pinds(\Mcal)\text{ and }[b,\pi(b)]\ni 0}=\emptyset$
is obviously a basis.
So we can assume that for an $i\in[0,\filtcnt-1]$, 
$\bigSet{\aG^b_i\given b\in\pinds(\Mcal)\text{ and }[b,\pi(b)]\ni i}$ 
is a basis of $\vsp_i$.
We have the following cases:
\begin{description}
    \item[$\psi_i$ an isomorphism:]
    In this case,
    $i\not\in\ninds(\Mcal)$ and $i+1\not\in\pinds(\Mcal)$.
    If $\psi_i:\vsp_i\to\vsp_{i+1}$ is forward, then
    $\bigSet{\aG^b_{i+1}\given b\in\pinds(\Mcal)\text{ and }[b,\pi(b)]\ni i+1}
    =\bigSet{\psi_i(\aG^b_{i})\given b\in\pinds(\Mcal)\text{ and }[b,\pi(b)]\ni i}$.
    % so $\bigSet{b\in\pinds_\Dim(\Mcal)\given[b,\pi(b)]\ni i+1}
    % =\bigSet{b\in\pinds_\Dim(\Mcal)\given[b,\pi(b)]\ni i}$.
    The elements in
    $\bigSet{\aG^b_{i+1}\given b\in\pinds(\Mcal)\text{ and }[b,\pi(b)]\ni i+1}$
    must then form a basis of $\vsp_{i+1}$
    because $\psi_i$ is an isomorphism.
    The verification for $\psi_i$ being backward is similar.
    
    \item[$\psi_i:\vsp_i\to\vsp_{i+1}$ forward, $\coker(\psi_i)$ non-trivial:]
    In this case,
    $i\not\in\ninds(\Mcal)$ and $i+1\in\pinds(\Mcal)$.
    % so $\bigSet{b\in\pinds_\Dim(\Mcal)\given[b,\pi(b)]\ni i+1}
    % =\bigSet{b\in\pinds_\Dim(\Mcal)\given[b,\pi(b)]\ni i}\union\bigSet{i+1}$.
    % Because $\morph{\Mcal}{i}{\Dim}$ {\red is injective}\footnote{
    %   \red Need to include the five possible forms of $\morph{\Mcal}{i}{\Dim}$
    %   in the preliminary.},
    % $\rank\Hm_\Dim(K_{i+1})=\rank\Hm_\Dim(K_{i})+1$, 
    % and $\aG^{i+1}_{i+1}\not\in\img(\morph{\Mcal}{i}{\Dim})$,
    % $\bigSet{\aG^b_{i+1}\given b\in\pinds_\Dim(\Mcal)\text{ and }[b,\pi(b)]\ni i+1}$ 
    % must be a basis of $\Hm_\Dim(K_{i+1})$.
    For each $b\in\pinds(\Mcal)$ such that $[b,\pi(b)]\ni i$,
    $[b,\pi(b)]\ni i+1$ and
    $\aG^b_i\mapsto \aG^b_{i+1}$ by $\psi_i$.
    We then have that elements in 
    $\bigSet{\aG^b_{i+1}=\psi_i(\aG^b_i)
    \given b\in\pinds(\Mcal)\text{ and }[b,\pi(b)]\ni i}$
    are linearly independent 
    because $\psi_i$ is injective.
    Since $\aG^{i+1}_{i+1}\not\in\img(\psi_i)$ 
    by Definition~\ref{dfn:rep-cls},
    $\bigSet{\aG^b_{i+1}\given b\in\pinds(\Mcal)\text{ and }[b,\pi(b)]\ni i+1}=
    \bigSet{\aG^b_{i+1}
    \given b\in\pinds(\Mcal)\text{ and }[b,\pi(b)]\ni i}
    \union\bigSet{\aG^{i+1}_{i+1}}$ must contain linearly independent elements.
    The fact that the cardinality of the set equals $\rank\vsp_{i+1}$ implies that
    it must form a basis of $\vsp_{i+1}$.

    \item[$\psi_i:\vsp_i\to\vsp_{i+1}$ forward, $\ker(\psi_i)$ non-trivial:]
    In this case,
    $i\in\ninds(\Mcal)$ and $i+1\not\in\pinds(\Mcal)$.
    Let $j=\pi\inv(i)$.
    For each $b\in\pinds(\Mcal)$ such that $[b,\pi(b)]\ni i$ and $b\neq j$,
    $[b,\pi(b)]\ni i+1$ and 
    $\aG^b_i\mapsto \aG^b_{i+1}$ by $\psi_i$.
    We then have that
    $\bigSet{\aG^b_{i+1}\given b\in\pinds(\Mcal)\text{ and }[b,\pi(b)]\ni i+1}
    =\bigSet{\psi_i\big(\aG^b_{i}\big)\given 
    b\in\pinds(\Mcal)\text{ and }[b,\pi(b)]\ni i}\setminus
    \bigSet{\psi_i\big(\aG_{i}^j\big)}$.
    % We have $\bigSet{b\in\pinds_\Dim(\Mcal)\given[b,\pi(b)]\ni i+1}
    % =\bigSet{b\in\pinds_\Dim(\Mcal)\given[b,\pi(b)]\ni i}\setminus\bigSet{j}$.
    Since $\psi_i$ is surjective,
    elements in $\bigSet{\psi_i\big(\aG_{i}^b\big)
    \given b\in\pinds(\Mcal)\text{ and }[b,\pi(b)]\ni i}$
    generate $\vsp_{i+1}$,
    in which $\psi_i\big(\aG_{i}^j\big)=0$ by Definition~\ref{dfn:rep-cls}.
    It follows that 
    $\bigSet{\aG^b_{i+1}\given b\in\pinds(\Mcal)\text{ and }[b,\pi(b)]\ni i+1}$
    forms a basis of $\vsp_{i+1}$
    because it generates $\vsp_{i+1}$
    and its cardinality equals $\rank\vsp_{i+1}$.

    \item[$\psi_i:\vsp_i\leftarrow\vsp_{i+1}$ backward, $\coker(\psi_i)$ non-trivial:]
    In this case,
    $i\in\ninds(\Mcal)$ and $i+1\not\in\pinds(\Mcal)$.
    % Let $j=\pi\inv(i)$.
    % We have $\bigSet{b\in\pinds_\Dim(\Mcal)\given[b,\pi(b)]\ni i+1}
    % =\bigSet{b\in\pinds_\Dim(\Mcal)\given[b,\pi(b)]\ni i}\setminus\bigSet{j}$.
    For each $b\in\pinds(\Mcal)$ such that $[b,\pi(b)]\ni i$ and $\pi(b)\neq i$,
    $[b,\pi(b)]\ni i+1$ and
    $\aG^b_i\mapsfrom \aG^b_{i+1}$ by $\psi_i$.
    We then have that elements in 
    $\bigSet{\aG^b_{i+1}=(\psi_i)\inv(\aG^b_i)
    \given b\in\pinds(\Mcal),\,[b,\pi(b)]\ni i,\,\text{and}\,\pi(b)\neq i}$
    are linearly independent 
    because if they are not,
    then their images under $\psi_i$ are also not,
    which is a contradiction.
    Note that 
    $\bigSet{\aG^b_{i+1}\given b\in\pinds(\Mcal),[b,\pi(b)]\ni i+1}=
    \bigSet{\aG^b_{i+1}
    \given b\in\pinds(\Mcal),\,[b,\pi(b)]\ni i,\,\text{and}\,\pi(b)\neq i}$
    and its cardinality equals $\rank\vsp_{i+1}$,
    so it must form a basis of $\vsp_{i+1}$.
    
    \item[$\psi_i:\vsp_i\leftarrow\vsp_{i+1}$ backward, $\ker(\psi_i)$ non-trivial:]
    In this case,
    $i\not\in\ninds(\Mcal)$ and $i+1\in\pinds(\Mcal)$.
    % so $\bigSet{b\in\pinds_\Dim(\Mcal)\given[b,\pi(b)]\ni i+1}
    % =\bigSet{b\in\pinds_\Dim(\Mcal)\given[b,\pi(b)]\ni i}\union\bigSet{i+1}$.
    For each $b\in\pinds(\Mcal)$ such that $[b,\pi(b)]\ni i$,
    $[b,\pi(b)]\ni i+1$ and
    $\aG^b_i\mapsfrom \aG^b_{i+1}$ by $\psi_i$.
    We then have that elements in 
    $\bigSet{\aG^b_{i+1}\in(\psi_i)\inv(\aG^b_i)
    \given b\in\pinds(\Mcal)\text{ and }[b,\pi(b)]\ni i}$
    are linearly independent 
    because their images under $\psi_i$ are.
    We also have that there is no non-trivial linear combination of
    $\bigSet{\aG^b_{i+1}
    \given b\in\pinds(\Mcal)\text{ and }[b,\pi(b)]\ni i}$
    falling in $\ker(\psi_i)$ because otherwise
    their images under $\psi_i$ would not be linearly 
    independent.
    Since $\aG_{i+1}^{i+1}$ is the non-zero element in $\ker(\psi_i)$
    by Definition~\ref{dfn:rep-cls},
    we have that 
    $\bigSet{\aG^b_{i+1}\given b\in\pinds(\Mcal)\text{ and }[b,\pi(b)]\ni i+1}=
    \bigSet{\aG^b_{i+1}
    \given b\in\pinds(\Mcal)\text{ and }[b,\pi(b)]\ni i}
    \union\bigSet{\aG_{i+1}^{i+1}}$
    contains linearly independent elements.
    Then, it must form a basis of $\vsp_{i+1}$
    because its cardinality equals $\rank\vsp_{i+1}$,
    \qedhere
\end{description}
\end{proof}

Now we present several propositions leading to 
our conclusion (Theorem~\ref{thm:0-zigzag-alg-correct}).
Specifically, 
% Proposition~\ref{prop:pn-paring-w-rep} states that
% if we have a pairing of positive and negative indices
% so that each induced interval admits representatives,
% then the induced intervals are the persistence intervals.
% Then,
Proposition~\ref{prop:path-induce-partial-rep} states that a certain path
in $\bcgraph(\Fcal)$ induces a set of partial 0-representatives.
% , and
Proposition~\ref{prop:0-zigzag-alg-invari} lists some invariants
of Algorithm~\ref{alg:0-zigzag}.
Proposition~\ref{prop:path-induce-partial-rep} and~\ref{prop:0-zigzag-alg-invari} 
support the proof of Proposition~\ref{prop:has-rep},
which together with Proposition~\ref{prop:pn-paring-w-rep} 
implies Theorem~\ref{thm:0-zigzag-alg-correct}.

From now on, $G$ and $\Fcal$ always denote the input to Algorithm~\ref{alg:0-zigzag}.
Since each node in a barcode graph
represents a connected component,
we also interpret nodes in a barcode graph as 0-th homology classes
throughout the paper.
Moreover, 
% for $j\leq i$,
a path in a barcode graph from a node $v$ to a node $u$ 
is said to be {\it within level $j$ and $i$} 
if for each node on the path, its level $\ell$
satisfies $j\leq \ell\leq i$;
we denote such a path as $\pathwl{v}{u}{j}{i}$.

\begin{proposition}\label{prop:path-induce-partial-rep}
Let $v$ be a level-$j$ node and $u$ be a level-$i$ node in $\bcgraph(\Fcal)$ 
such that $j<i$ and there is a path $\pathwl{v}{u}{j}{i}$ in $\bcgraph(\Fcal)$.
% \sout{Then, the path  $\pathwl{v}{u}{j}{i}$ induces}
Then, there is
a set of partial 0-representatives $\Set{\aG_k\in\Hm_0(G_k)\given k\in[j,i]}$
for the interval $[j,i]$
with $\aG_j=v$ and $\aG_i=u$.
\end{proposition}
\begin{proof}
We can assume that $\pathwl{v}{u}{j}{i}$ is a simple path because if it were not
we could always find one.
% from $v$ to $u$ within level $j$ and $i$,
For each $k\in[j+1,i-1]$, 
let $w_1,\ldots,w_r$ be all the level-$k$ nodes on $\pathwl{v}{u}{j}{i}$
whose adjacent nodes on $\pathwl{v}{u}{j}{i}$ are at different levels.
% \tamal{just curious, why didn't you say excluding splitting and merging nodes?}
Then, let 
$\aG_k=\sum_{\ell=1}^r w_\ell$.
% $\aG_k=w_1+\cdots+w_\ell$.
Also, let $\aG_j=v$ and $\aG_i=u$.
It can be verified that $\Set{\aG_k\given k\in[j,i]}$
is a set of partial 0-representatives for $[j,i]$.
See Figure~\ref{fig:path} for an example of a simple path 
$\pathwl{\tilde{v}_2}{u_2}{10}{13}$ (the dashed one)
in a barcode graph, where
the solid nodes 
% have the adjacent nodes at different levels and therefore 
contribute to the induced partial 0-representatives
and the hollow nodes are excluded.
\end{proof}

For
% the filtration $\Fcal$ in Algorithm~\ref{alg:0-zigzag} and 
an $i$ with $0\leq i\leq\filtcnt$,
we define the {\it prefix} $\Fcal^i$ of $\Fcal$ as the filtration 
$\Fcal^i:G_0\leftrightarrow\cdots\leftrightarrow G_i$
% We note that for a {\red prefix} $\Fcal^i$ of $\Fcal$, 
% the barcode graph 
and observe that $\bcgraph(\Fcal^i)$ is the subgraph of $\bcgraph(\Fcal)$
induced by nodes at levels less than or equal to $i$.
% From now on, 
We call level-$i$ nodes of $\bcgraph(\Fcal^i)$ as {\it leaves}
and do not distinguish 
leaves in $\bcforest_i$ and $\bcgraph(\Fcal^i)$
because they bijectively map to each other.
It should be clear from the context though which graph or forest 
a particular leaf is in.

\begin{figure}
  \centering
  \includegraphics[width=0.85\linewidth]{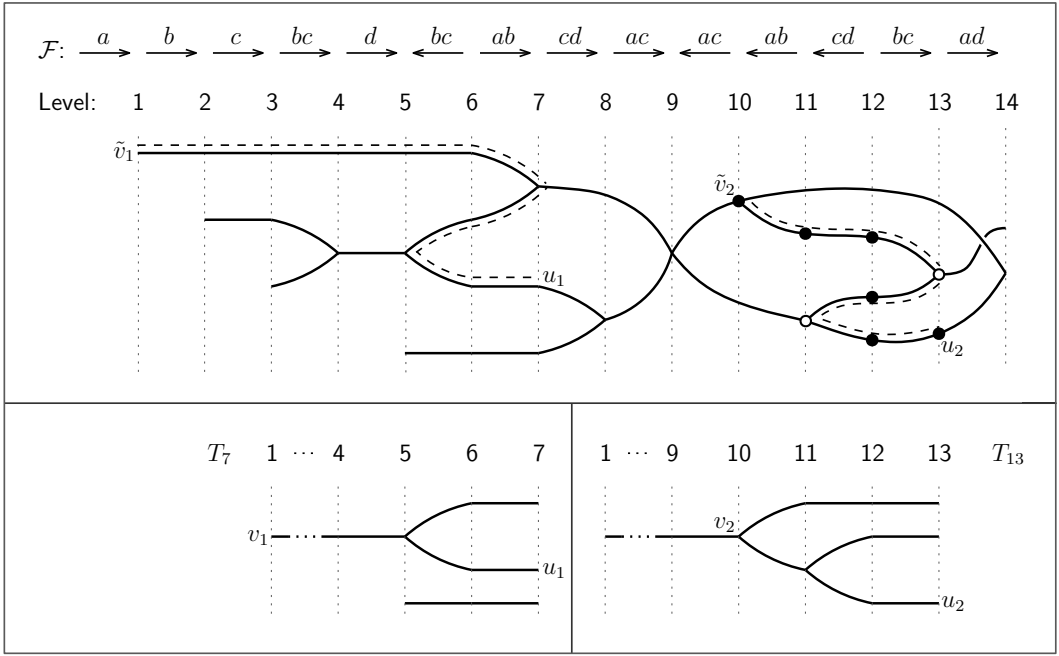}
  \vspace{1em}
  \caption{Illustration of invariants of Proposition~\ref{prop:0-zigzag-alg-invari}.
  The top part contains a barcode graph with its filtration given
  ($a$, $b$, $c$, and $d$
  are vertices of the complex). The bottom contains two barcode forests.}
  \label{fig:path}
\end{figure}

\begin{proposition}\label{prop:0-zigzag-alg-invari}
For each $i=0,\ldots,\filtcnt$, 
Algorithm~\ref{alg:0-zigzag} maintains the following invariants:
\begin{enumerate}
    \item\label{itm:tree-graph-bij-invari}
    There is a bijection $\eta$ from 
    trees in $\bcforest_i$ to 
    connected components in $\bcgraph(\Fcal^i)$ containing leaves
    such that a leaf $u$ is in a tree $\UG$ of $\bcforest_i$ if and only if
    $u$ is in $\eta(\UG)$. 
    
    \item\label{itm:ances-leaf-path-invari}
    For each leaf $u$ in $\bcforest_i$
    and each ancestor of $u$ at a level $j$,
    there is a path $\pathwl{\tilde{v}}{u}{j}{i}$
    % from a level-$j$ node to $u$
    % within level $j$ and $i$ 
    in $\bcgraph(\Fcal)$
    where $\tilde{v}$ is a level-$j$ node.
    
    % For each tree $\UG$ in $\bcforest_i$, each leaf $u$ in $\UG$,
    % and each ancestor $v$ of $u$ at a level $j$,
    % there is a path from a level-$j$ node to $u$
    % within level $j$ and $i$ in $\eta(\UG)$.

    \item\label{itm:split-ances-leaf-path-invari}
    For each leaf $u$ in $\bcforest_i$
    and each splitting ancestor of $u$ at a level $j$,
    let $\tilde{v}$ be the unique level-$j$ splitting node in $\bcgraph(\Fcal)$.
    Then, there is a path $\pathwl{\tilde{v}}{u}{j}{i}$
    % from $w$ to $u$ within level $j$ and $i$ 
    in $\bcgraph(\Fcal)$.

    % For each tree $\UG$ in $\bcforest_i$, each leaf $u$ in $\UG$,
    % and each splitting ancestor $v$ of $u$ at a level $j$,
    % let $w$ be the unique level-$j$ splitting node in $\eta(\UG)$.
    % Then, there is a path from $w$ to $u$ within level $j$ and $i$ in $\eta(\UG)$.
\end{enumerate}
\end{proposition}
\begin{remark}
See Figure~\ref{fig:path} for examples of invariant~\ref{itm:ances-leaf-path-invari}
and~\ref{itm:split-ances-leaf-path-invari}.
In the figure,
$v_1$ is a level-1 non-splitting ancestor of $u_1$ in $T_7$
and $\tilde{v}_1$ is a level-1 node in the barcode graph;
$v_2$ is a level-10 splitting ancestor of $u_2$ in $T_{13}$
and $\tilde{v}_2$ is the unique level-10 splitting node in the barcode graph. 
The paths $\pathwl{\tilde{v}_1}{u_1}{1}{7}$ 
and $\pathwl{\tilde{v}_2}{u_2}{10}{13}$ are marked with dashes.
\end{remark}
% \begin{proof}
% See Appendix~\ref{sec:pf-prop-0-zigzag-alg-invari}.
% \end{proof}
\begin{proof}
We only verify invariant~\ref{itm:split-ances-leaf-path-invari} as 
the verification for invariant~\ref{itm:ances-leaf-path-invari}
is similar but easier and invariant~\ref{itm:tree-graph-bij-invari}
is straightforward.
The verification is by induction.
When $i=0$, invariant~\ref{itm:split-ances-leaf-path-invari} trivially holds.
Now suppose that invariant~\ref{itm:split-ances-leaf-path-invari} is true 
for an $i\in[0,\filtcnt-1]$.
For the no-change, entrance, and split event in Algorithm~\ref{alg:0-zigzag},
it is not hard to see that invariant~\ref{itm:split-ances-leaf-path-invari}
still holds for $i+1$.
For the departure event, 
because we are only deleting a path from $\bcforest_i$ to form $\bcforest_{i+1}$,
invariant~\ref{itm:split-ances-leaf-path-invari}
also holds for $i+1$.
For the merge event,
let 
% $\UG$ be a tree in $\bcforest_{i+1}$
$u$ be a leaf in $\bcforest_{i+1}$,
$v$ be a splitting ancestor of $u$ at a level $j$,
and $\tilde{v}$ be the unique splitting node in $\bcgraph(\Fcal)$ at level $j$.
The node $v$ may correspond to one or two nodes in $\bcforest_i$,
in which only one is splitting, and let $v'$ be the splitting one.
% Note that $v'$ is also a level-$j$ node.
Note that $u$'s parent 
may correspond to one or two nodes in $\bcforest_i$,
and we let $W$ be the set of nodes in $\bcforest_i$
that $u$'s parent corresponds to.
If $v'$ is an ancestor of a node $w\in W$
% \footnote{Note 
%   that $u$'s parent in $\bcforest_{i+1}$ may correspond to two nodes in $\bcforest_i$,
%   and we only need one of them to satisfy the condition.
%   Similar situations apply to other places of the proof.} 
in $\bcforest_i$, then by the assumption,
there must be a path $\pathwl{\tilde{v}}{w}{j}{i}$
% from $\tilde{v}$ to $w$ within level $j$ and $i$ 
in $\bcgraph(\Fcal)$.
From this path,
we can derive a path $\pathwl{\tilde{v}}{u}{j}{i+1}$
% from $w$ to $u$ within level $j$ and $i+1$ 
in $\bcgraph(\Fcal)$.
If $v'$ is not an ancestor of any node of $W$ in $\bcforest_i$, 
the fact that $v$ is an ancestor of $u$'s parent in $\bcforest_{i+1}$ implies that
there must be an ancestor $v''$ of a node $w\in W$ in $\bcforest_{i}$
which $v$ corresponds to.
So we have that $v$ is a gluing of two nodes from $\bcforest_{i}$.
Note that $u$'s parent must not be a glued node in $\bcforest_{i+1}$
because otherwise $v'$ would have been an ancestor 
of a node of $W$ in $\bcforest_i$;
see Figure~\ref{fig:path_exist}
where $z_1$ and $z_2$ are the two level-$i$  nodes glued together.
Let $x$ be the highest one among the nodes on the path from $v$ to $u$
that are glued in iteration~$i$.
We have that $x$ must correspond to a node $x'$ in $\bcforest_i$ 
which is an ancestor of $w$.
% If $y=v$, then $v$ corresponds to two nodes in $\bcforest_i$ one of which is $v'$,
% and let the other one be $v''$.
% Because $v$ is an ancestor of $u$ in $\bcforest_{i+1}$,
% one of $v'$ and $v''$ must be an ancestor of $u$'s parent in $\bcforest_{i}$,
% and so $v''$ is the ancestor of $u$'s parent.
Recall that $z_1,z_2$ are the two leaves in $\bcforest_{i}$ which are glued,
and let $z_3$ be the child of the glued node of $z_1,z_2$ in $\bcforest_{i+1}$,
as shown in Figure~\ref{fig:path_exist}.
From the figure,
we have that $x'$ must be splitting because 
one child of $x'$ (which is not glued) descends down to $w$
and the other child of $x'$ (which is glued) descends down to $z_1$.
The fact that $v'$ is an ancestor of $z_2$ in $\bcforest_i$ implies that 
there is a path $\pathwl{\tilde{v}}{z_2}{j}{i}$
% from $w$ to $z_2$ within level $j$ and $i$ 
in $\bcgraph(\Fcal)$.
Let $\tilde{x}$ be the unique splitting node in $\bcgraph(\Fcal)$ at the same level with $x'$;
then, $z_1$ and $w$ being descendants of $x'$ in $\bcforest_i$ 
implies that there are 
paths $\pathwl{z_1}{\tilde{x}}{j}{i}$ and $\pathwl{\tilde{x}}{w}{j}{i}$
% from $z_1$ to $\tilde{x}$ and from $\tilde{x}$ to $w$
% within level $j$ and $i$ 
in $\bcgraph(\Fcal)$.
Now
% concatenating the above mentioned paths 
% and the red edges in Figure~\ref{fig:path_exist},
we derive a path $\pathwl{\tilde{v}}{u}{j}{i+1}$
% from $w$ to $u$ within level $j$ and $i+1$ 
in $\bcgraph(\Fcal)$
by concatenating the following paths and edges:
$\pathwl{\tilde{v}}{z_2}{j}{i}$, 
$\bar{z_2z_3}$,
$\bar{z_3z_1}$,
$\pathwl{z_1}{\tilde{x}}{j}{i}$,
$\pathwl{\tilde{x}}{w}{j}{i}$,
$\bar{wu}$.
\end{proof}

\begin{figure}
  \centering
  \includegraphics[width=0.82\linewidth]{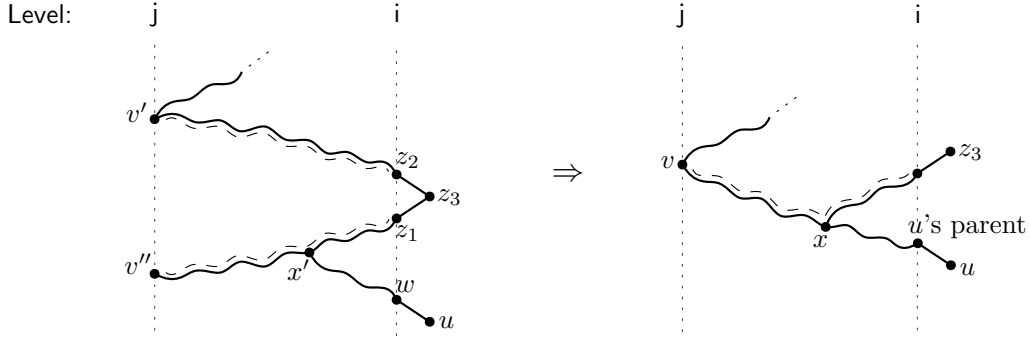}
  \vspace{1em}
  \caption{Illustration of parts of $T_{i+1}$ 
  for the proof of Proposition~\ref{prop:0-zigzag-alg-invari}, 
  where the left one is before the path gluing
  and the right one is after.
  Note that the part between level $j$ and $i$ for the left tree
  actually belongs to $T_i$. The paths with dashed marks are the glued ones
  (before and after), in which $v'$ and $v''$ are identified as $v$, 
  and $x'$ is identified as $x$ with another node.
%   The path with dashed mark in the right tree is 
%   the portion of the path from $v$ to $u$ which are glued.
  }
  \label{fig:path_exist}
\end{figure}

\begin{proposition}\label{prop:has-rep}
Each interval produced by Algorithm~\ref{alg:0-zigzag} admits a set of 0-representatives.
\end{proposition}
\begin{proof}
Suppose that an interval is produced by the merge event at iteration~$i$.
We have the following situations:
\begin{itemize}
    \item 
% If $u_1$,$u_2$ 
If the nodes $u_1,u_2$ in this event (see Algorithm~\ref{alg:0-zigzag})
are in the same tree in $\bcforest_i$,
let $v$ be the highest common ancestor of $u_1,u_2$
and note that $v$ is a splitting node at level $j$.
Also note that
$u_1,u_2$ are actually leaves in $\bcforest_i$ and hence
can also be considered as level-$i$ nodes in $\bcgraph(\Fcal)$.
% we can consider $u_1,u_2$ as level-$i$ nodes in $\bcgraph(\Fcal)$
% because leaves in $\bcforest_i$ and level-$i$ nodes in $\bcgraph(\Fcal)$
% bijectively map to each other.
Let $\tilde{v}$ be the unique level-$j$ splitting node in $\bcgraph(\Fcal)$.
By invariant~\ref{itm:split-ances-leaf-path-invari} 
of Proposition~\ref{prop:0-zigzag-alg-invari}
along with Proposition~\ref{prop:path-induce-partial-rep},
there are two sets of partial 0-representatives
$\Set{\aG_k\given k\in[j,i]},\Set{\bG_k\given k\in[j,i]}$ for $[j,i]$ 
with $\aG_j=\tilde{v}$, $\aG_i=u_1$, $\bG_j=\tilde{v}$, and $\bG_i=u_2$.
We claim that $\Set{\aG_k+\bG_k\given k\in[j+1,i]}$
is a set of 0-representatives
for the interval $[j+1,i]$.
To prove this, we first note the following obvious facts: 
{\sf(i)}~$\Set{\aG_k+\bG_k\given k\in[j+1,i]}$
is a set of partial 0-representatives;
{\sf(ii)}~$\aG_{j+1}+\bG_{j+1}\in\ker(\morph{\Fcal}{j}{0})$;
{\sf(iii)}~$\aG_{i}+\bG_{i}$ is the non-zero element in $\ker(\morph{\Fcal}{i}{0})$.
So we only need to show that $\aG_{j+1}+\bG_{j+1}\neq 0$.
Let $v_1,v_2$ be the two level-$(j+1)$ nodes in $\bcgraph(\Fcal)$
connecting to $\tilde{v}$.
Then, $\aG_{j+1}$ equals $v_1$ or $v_2$ and the same for $\bG_{j+1}$.
To see this, 
we first show that $\aG_{j+1}$ can only contain $v_1,v_2$.
For contradiction,
suppose instead that $\aG_{j+1}$ contains
a level-$(j+1)$ node $x$ with $x\neq v_1$, $x\neq v_2$. 
Let $\pathwl{\tilde{v}}{u_1}{j}{i}$ be the simple path 
that induces $\Set{\aG_k\given k\in[j,i]}$ as in
Proposition~\ref{prop:path-induce-partial-rep} and its proof.
Then, $x$ is on the path $\pathwl{\tilde{v}}{u_1}{j}{i}$ and
the two adjacent nodes of $x$ on $\pathwl{\tilde{v}}{u_1}{j}{i}$
are at level $j$ and $j+2$,
% (see the proof of Proposition~\ref{prop:path-induce-partial-rep}), 
in which we let $y$ be the one at level $j$.
Note that $y\neq\tilde{v}$ because 
$x$ is not equal to $v_1$ or $v_2$. 
Since $\pathwl{\tilde{v}}{u_1}{j}{i}$ is within level $j$ and $i$,
$y$ must be adjacent to another level-$(j+1)$ node distinct from $x$
on $\pathwl{\tilde{v}}{u_1}{j}{i}$.
Now we have that $y$ is a level-$j$ splitting node with $y\neq\tilde{v}$,
contradicting the fact that $\bcgraph(\Fcal)$ 
has only one level-$j$ splitting node.
The fact that $\aG_{j+1}$ contains $v_1$ or $v_2$ but not both
can be similarly verified.
To see that $\aG_{j+1}+\bG_{j+1}\neq 0$,
suppose instead that $\aG_{j+1}+\bG_{j+1}= 0$, i.e., $\aG_{j+1}=\bG_{j+1}$, and
without loss of generality they both equal $v_1$. 
% Then, there are paths $\pathwl{v_1}{u}{j}{i}$ and $\pathwl{v_1}{w}{j}{i}$
% from $v_1$ to $u$ and $w$ within level $j$ and $i$ 
% in $\bcgraph(\Fcal)$.
Note that we can consider $T_i$ as derived by contracting nodes of $\bcgraph(\Fcal^i)$ 
at the same level\footnote{
  We should further note that 
  this contraction is not done on the entire $\bcgraph(\Fcal^i)$ 
  but rather on connected components of $\bcgraph(\Fcal^i)$ containing leaves.}.
The fact that $\aG_{j+1}=\bG_{j+1}=v_1$
implies that $u_1,u_2$ are descendants
of the same child of $v$ in $T_i$,
contradicting the fact that $v$ is the highest common ancestor of $u_1,u_2$.
So we have that 
% $\aG_{j+1}\neq\bG_{j+1}$ and hence 
$\aG_{j+1}+\bG_{j+1}\neq 0$. 

    \item
% If the nodes $u_1$,$u_2$ in this case (see Algorithm~\ref{alg:0-zigzag}) 
If $u_1,u_2$
are in different trees in $\bcforest_i$,
then without loss of generality 
let $u_1$ be the one whose root $v_1$ is at the higher level (i.e., level $j$).
% {\gray Also, let $v_2$ be a level-$j$ ancestor of $u_2$}.
% {\gray Note that we can consider $u_1$ and $u_2$ as level-$i$ nodes in $\bcgraph(\Fcal)$
% because leaves in $\bcforest_i$ and level-$i$ nodes in $\bcgraph(\Fcal)$
% bijectively map to each other.}
As the root of $u_1$,
the node $v_1$ must be an entering node, and
the connected component of $\bcgraph(\Fcal^i)$ containing $u_1$
must have a single level-$j$ node $\tilde{v}_1$.
Then, by invariant~\ref{itm:ances-leaf-path-invari}
of Proposition~\ref{prop:0-zigzag-alg-invari}
along with Proposition~\ref{prop:path-induce-partial-rep},
there are two sets of partial 0-representatives
$\Set{\aG_k\given k\in[j,i]},\Set{\bG_k\given k\in[j,i]}$ for $[j,i]$ 
with $\aG_j=\tilde{v}_1$, $\aG_i=u_1$, $\bG_j=\tilde{v}_2$, and $\bG_i=u_2$,
where $\tilde{v}_2$ is a level-$j$ node.
We claim that $\Set{\aG_k+\bG_k\given k\in[j,i]}$
is a set of 0-representatives
for the interval $[j,i]$
and the verification is similar to the previous 
case where $u_1$ and $u_2$ are in the same tree.
\end{itemize}

For intervals produced by the departure events and 
at the end of the algorithm,
the existence of 0-representatives can be similarly argued.
% E.g., when the interval $[11,13]$ is produced by the merging case
% for the example shown in Figure~\ref{fig:path},
% the two paths within level $10$ and $14$ from $\tilde{v}_2$ to $\gG$
% induce the 0-representatives.
\end{proof}

\begin{theorem}\label{thm:0-zigzag-alg-correct}
Algorithm~\ref{alg:0-zigzag} computes the 0-th zigzag barcode
for a given zigzag filtration.
\end{theorem}
\begin{proof}
First, we have the following facts: 
every level-$j$ entering node in $\bcgraph(\Fcal)$
introduces a $j\in\pinds(\Hm_0(\Fcal))$ 
and uniquely corresponds to a level-$j$ root in $\bcforest_i$ for some $i$;
every level-$j$ splitting node in $\bcgraph(\Fcal)$
introduces a $j+1\in\pinds(\Hm_0(\Fcal))$ 
and uniquely corresponds to a level-$j$ splitting node in $\bcforest_i$ for some $i$.
Whenever an interval $[j,i]$ is produced in Algorithm~\ref{alg:0-zigzag},
$i\in\ninds(\Hm_0(\Fcal))$ and the entering or splitting node in $\bcforest_i$ 
introducing $j$ as a positive index either becomes a {\it regular} node 
(i.e., connecting to a single node on both adjacent levels)
or is deleted in $\bcforest_{i+1}$.
This means that $j$ is never the start of another interval produced.
At the end of Algorithm~\ref{alg:0-zigzag}, the number of intervals produced
which end with $\filtcnt$
also matches the rank of $\Hm_0(G_\filtcnt)$.
Therefore, intervals produced by the algorithm
induce a bijection $\pi:\pinds(\Hm_0(\Fcal))\to\ninds(\Hm_0(\Fcal))$.
By Proposition~\ref{prop:pn-paring-w-rep} and~\ref{prop:has-rep},
our conclusion follows.
% we only need to show that each interval
% admits a set of 0-representatives.
\end{proof}

\section{One-dimensional zigzag persistence}
\label{sec:1-zigzag}
In this section, we present an efficient algorithm for 1-st zigzag persistence
% (see Algorithm~\ref{alg:1-zigzag-graph}) 
on graphs.
We assume that
the input is a graph $G$ with a zigzag filtration
\[
% $
\Fcal:\emptyset=G_0 \leftrightarrowsp{\fsimp{\Fcal}{0}} 
G_1 \leftrightarrowsp{\fsimp{\Fcal}{1}} \cdots 
\leftrightarrowsp{\fsimp{\Fcal}{\filtcnt-1}}  G_\filtcnt
% $
\]
of $G$.
% To better illustrate the idea,
We first describe the algorithm 
%with an emphasis on its structure
without giving the full implementation details.
The key  to the algorithm 
is a pairing principle
for the positive and negative indices.
We then prove the correctness of the algorithm.
Finally, in Section~\ref{sec:1-zigzag-impl}, we make several observations 
which reduce the index pairing to 
finding the {\it max edge-weight}
of a path in a minimum spanning forest,
leading to an efficient implementation.
% Using a data structure~\cite{holm2001poly} 
% for dynamically maintaining the minimum spanning forest (MSF),
% the abstract algorithm can be implemented efficiently.

% Before presenting the algorithm,
We notice that the following are true
for every inclusion $G_{i}\leftrightarrowsp{\fsimp{\Fcal}{i}} G_{i+1}$ of $\Fcal$
(recall that $\morph{\Fcal}{i}{1}$ denotes
the corresponding linear map in the induced module $\Hm_1(\Fcal)$):
\begin{itemize}
    \item If $\fsimp{\Fcal}{i}$ is an edge being added
    and vertices of $\fsimp{\Fcal}{i}$ are connected
    in $G_i$,
    then $\morph{\Fcal}{i}{1}$ is an injection with non-trivial cokernel,
    which provides $i+1\in\pinds(\Hm_1(\Fcal))$.
    
    \item If $\fsimp{\Fcal}{i}$ is an edge being deleted
    and vertices of $\fsimp{\Fcal}{i}$ are connected
    in $G_{i+1}$,
    then $\morph{\Fcal}{i}{1}$ is an injection with non-trivial cokernel,
    which provides $i\in\ninds(\Hm_1(\Fcal))$.
    
    \item In all the other cases, 
    $\morph{\Fcal}{i}{1}$ is an isomorphism
    and $i\not\in\ninds(\Hm_1(\Fcal))$, $i+1\not\in\pinds(\Hm_1(\Fcal))$.
\end{itemize}

% As mentioned in Section~\ref{sec:0-zigzag-proof},
As can be seen from Section~\ref{sec:0-zigzag},
computing $\Pers(\Hm_1(\Fcal))$ boils down to finding a pairing of
indices of $\pinds(\Hm_1(\Fcal))$ and $\ninds(\Hm_1(\Fcal))$.
Our algorithm adopts this structure,
where $\upos_i$ denotes 
the set of unpaired positive indices at the {\it beginning} of each iteration $i$:

\pagebreak

\begin{algr}[Algorithm for 1-st zigzag persistence on graphs]
\label{alg:1-zigzag-graph-abstr}
{\begin{itemize}
  \item[]
  \item[] $\upos_0:=\emptyset$
  \item[] \textbf{for} $i:= 0,\ldots,\filtcnt-1${\rm:}
  \begin{itemize}
      \item[] \textbf{if} $G_{i}\rightarrowsp{\fsimp{\Fcal}{i}} G_{i+1}$
      provides $i+1\in\pinds(\Hm_1(\Fcal))${\rm:}
    %   \hfill
    %   {\sf$\triangleright$ $\fsimp{\Fcal}{i}$ must be an edge}
      \begin{itemize}
          \item[] 
          $\upos_{i+1}:=\upos_i\union\Set{i+1}$
        %   add $i+1$ to the set of unpaired positive indices
      \end{itemize}
      \item[] \textbf{else if} $G_{i}\leftarrowsp{\fsimp{\Fcal}{i}} G_{i+1}$
      provides $i\in\ninds(\Hm_1(\Fcal))${\rm:}
    %   \hfill
    %   {\sf$\triangleright$ $\fsimp{\Fcal}{i}$ must be an edge}
      \begin{itemize}
          \item[] pair $i$ with a $j_*\in\upos_i$ based on the Pairing Principle below
          \item[] output an interval $[j_*,i]$ for $\Pers(\Hm_1(\Fcal))$
          \item[] $\upos_{i+1}:=\upos_i\setminus\Set{j_*}$
        %   delete $j_*$ from the set of unpaired positive indices
      \end{itemize}
      \item[] \textbf{else}{\rm:}
      \begin{itemize}
          \item[] $\upos_{i+1}:=\upos_i$
      \end{itemize}
  \end{itemize}
  \item[] \textbf{for each} $j\in\upos_\filtcnt${\rm:}
  \begin{itemize}
      \item[] output an interval $[j,\filtcnt]$ for $\Pers(\Hm_1(\Fcal))$
  \end{itemize}
\end{itemize}}
\end{algr}

Note that in Algorithm~\ref{alg:1-zigzag-graph-abstr},
whenever a positive or negative index is produced, 
$\fsimp{\Fcal}{i}$ must be an edge.
One key piece missing from the algorithm
% Algorithm~\ref{alg:1-zigzag-graph-abstr}
is how we choose a positive index to pair with a negative index:
\theoremstyle{theorem}
\newtheorem*{pprincile}{Pairing Principle for Algorithm~\ref{alg:1-zigzag-graph-abstr}}
% \begin{theorem}[Pairing theorem for 1-st zigzag on graphs]
% \label{thm:pairing}
% In each iteration $i$ of Algorithm~\ref{alg:1-zigzag-graph-abstr} 
% such that $G_{i}\leftrightarrowsp{\fsimp{\Fcal}{i}} G_{i+1}$
% produces an $i\in\ninds(\Hm_1(\Fcal))$,
% let $J_i$ consist of any unpaired positive index $j$ such that
% there exists a 1-cycle $z$ with {\red$z\subseteq G_k$} for each $k\in[j,i]$,
% $\fsimp{\Fcal}{j-1}\in z$, and $\fsimp{\Fcal}{i}\in z$.
% Then, {\red$J_i\neq\emptyset$} and 
% $[j_*,i]$ forms a persistence interval of $\Pers(\Hm_1(\Fcal))$,
% where $j_*$ is the least index in $J_i$.
% \end{theorem}
\begin{pprincile}
In each iteration $i$ %of Algorithm~\ref{alg:1-zigzag-graph-abstr} 
where $G_{i}\leftarrowsp{\fsimp{\Fcal}{i}} G_{i+1}$
provides $i\in\ninds(\Hm_1(\Fcal))$,
let $J_i$ consist of every $j\in\upos_i$ such that
there exists a 1-cycle $z$ containing both $\fsimp{\Fcal}{j-1}$ and $\fsimp{\Fcal}{i}$
with $z\subseteq G_k$ for every $k\in[j,i]$.
Then, $J_i\neq\emptyset$ and
Algorithm~\ref{alg:1-zigzag-graph-abstr} pairs $i$ 
with the {\rm smallest} index {$j_*$} in $J_i$.
\end{pprincile}
\begin{remark}
See Proposition~\ref{prop:iter-cyc-exist}
% \tao{presented later}
% {\red in Section~\ref{sec:1-zigzag-proof}} 
for a proof of $J_i\neq\emptyset$ claimed above.
\end{remark}
\begin{remark}
Algorithms for non-zigzag 
persistence~\cite{edelsbrunner2000topological,zomorodian2005computing} always 
pair a negative index with the \textit{largest} (i.e., youngest) positive index
satisfying a certain condition, 
while Algorithm~\ref{alg:1-zigzag-graph-abstr} pairs with the smallest one.
This is due to the difference of zigzag and non-zigzag persistence
and our particular condition that 1-cycles can never become boundaries
in graphs.
See~\cite{carlsson2009zigzag-realvalue,maria2014zigzag}
for the pairing when assuming general zigzag filtrations.
\end{remark}

\begin{figure}
  \centering
  \includegraphics[width=\linewidth]{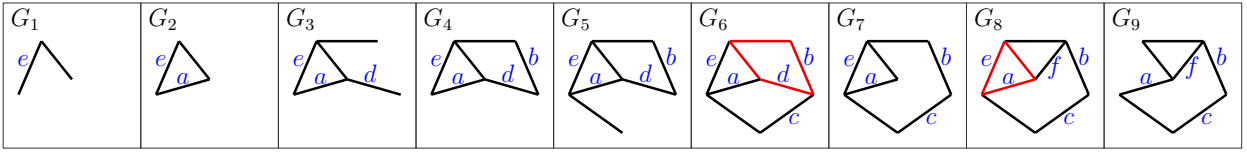}
  \caption{A zigzag filtration with 1-st barcode $\Set{[4,6],[2,8],[6,9],[8,9]}$. 
  For brevity, the addition of vertices and some edges are skipped.}
  \label{fig:1cyc}
\end{figure}

% We postpone the proof of Theorem~\ref{thm:pairing} to Section~\ref{sec:1-zigzag-proof}.
Figure~\ref{fig:1cyc} gives an example of the pairing of the indices
and their corresponding edges.
% where the edges creating or killing 1-cycles are labelled.
In the figure, when edge $d$ is deleted from $G_6$, 
there are three unpaired positive edges $a$, $b$, and $c$,
in which $b$ and $c$ admit 1-cycles as required by the Pairing Principle.
As the earlier edge, $b$ is paired with $d$
and an interval $[4,6]$ is produced.
The red cycle in $G_6$ indicates the 1-cycle containing $b$ and $d$
which exists in all the intermediate graphs.
Similar situations happen when $e$ is paired with $a$ in $G_8$,
producing the interval $[2,8]$.
% {\red We prove the correctness of Algorithm~\ref{alg:1-zigzag-graph-abstr} 
% in Section~\ref{sec:1-zigzag-proof}.}
% \tao{We postpone the justification of Algorithm~\ref{alg:1-zigzag-graph-abstr} 
% to the end of this section.}

For the correctness of Algorithm~\ref{alg:1-zigzag-graph-abstr},
we first provide 
% the following proposition
Proposition~\ref{prop:iter-cyc-exist} 
which justifies
the Pairing Principle
and is a major step leading toward our conclusion (Theorem~\ref{thm:1-zigzag-graph-corr}):

\begin{proposition}\label{prop:iter-cyc-exist}
At the beginning of each iteration $i$ in Algorithm~\ref{alg:1-zigzag-graph-abstr}, 
for every $j\in\upos_i$,
there exists a 1-cycle $z^i_j$ containing $\fsimp{\Fcal}{j-1}$ 
with $z^i_j\subseteq G_k$ for every $k\in[j,i]$.
Furthermore, the set $\Set{z^i_j\given j\in\upos_i}$ forms a basis of $\Zyc_1(G_i)$.
If the iteration $i$ produces a negative index $i$, 
then the above statements imply that 
there is at least one $z^i_j$ containing $\fsimp{\Fcal}{i}$.
This $z^i_j$ satisfies the condition that $z^i_j\subseteq G_k$ for every $k\in[j,i]$,
$\fsimp{\Fcal}{j-1}\in z^i_j$, and $\fsimp{\Fcal}{i}\in z^i_j$,
which implies that $J_i\neq\emptyset$ where $J_i$ is as defined in the Pairing Principle.
\end{proposition}
% \begin{proof}
% See Appendix~\ref{sec:pf-prop-iter-cyc-exist}.
% \end{proof}
\begin{proof}
We prove this by induction. At the beginning of iteration $0$,
since $G_0=\emptyset$ and $\upos_0=\emptyset$,
the proposition is trivially true.
Suppose that the proposition is true at the beginning of an iteration $i$.
For each $j\in\upos_i$, 
let $z^i_j$ be the 1-cycle as specified in the proposition.
If $G_{i}\leftrightarrowsp{\fsimp{\Fcal}{i}} G_{i+1}$
produces neither a positive index nor a negative index,
then $\Zyc_1(G_i)=\Zyc_1(G_{i+1})$ and 
$\upos_i=\upos_{i+1}$. 
Let $z^{i+1}_j=z^i_j$ for each $j$;
then, $\Set{z^{i+1}_j\given j\in\upos_{i+1}}$ 
serves as the 1-cycles as specified in the proposition
for iteration $i+1$.
If $G_{i}\rightarrowsp{\fsimp{\Fcal}{i}} G_{i+1}$
produces a new unpaired positive index $i+1$,
let $z^{i+1}_{i+1}$ be any 1-cycle in $G_{i+1}$ containing $\fsimp{\Fcal}{i}$.
Also, for each $j\in\upos_i$,
let $z^{i+1}_j=z^i_j$.
It can be verified that 
$\Set{z^{i+1}_j\given j\in\upos_{i+1}}$ serves as the 1-cycles as specified in the proposition
for iteration $i+1$.

If $G_{i}\leftarrowsp{\fsimp{\Fcal}{i}} G_{i+1}$
produces a negative index $i$, 
then there must be a 1-cycle in $G_i$ containing $\fsimp{\Fcal}{i}$.
The fact that $\Set{z^{i}_j\given j\in\upos_{i}}$ forms a basis of $\Zyc_1(G_i)$
implies that there must be at least one $z^i_j$ containing $\fsimp{\Fcal}{i}$
because otherwise no combination of the $z^i_j$'s can 
equal a cycle containing $\fsimp{\Fcal}{i}$.
Let $\bar{j}$ be the smallest $j\in\upos_i$
such that $z^i_j$ contains $\fsimp{\Fcal}{i}$.
We claim that $j_*=\bar{j}$, where $j_*$ is as defined in the Pairing Principle.
For contradiction, suppose instead that $j_*\neq\bar{j}$.
% If $j_*>\bar{j}$, then it can be verified that 
Note that $\bar{j}\in J_i$, where $J_i$ is as defined in the Pairing Principle.
% This contradicts the fact that 
Since $j_*$ is the smallest index in $J_i$,
we have that $j_*<\bar{j}$.
By the Pairing Principle,
there exists a 1-cycle $\zG$ containing both 
$\fsimp{\Fcal}{j_*-1}$ and $\fsimp{\Fcal}{i}$
with $\zG\subseteq G_k$ for every $k\in[j_*,i]$.
Since $\Set{z^{i}_j\given j\in\upos_i}$ forms a basis of $\Zyc_1(G_i)$
and $\zG\subseteq G_i$,
$\zG$ must equal a sum $\sum_{\ell=1}^s z^i_{\lG_\ell}$,
where each $\lG_\ell\in\upos_i$.
We rearrange the indices such that
$\lG_1<\lG_2<\cdots<\lG_s$.
We have that $\lG_s\geq\bar{j}$ because otherwise
each $\lG_\ell<\bar{j}$ and 
so its corresponding $z^i_{\lG_\ell}$ does not contain $\fsimp{\Fcal}{i}$.
This implies that $\zG=\sum_{\ell=1}^s z^i_{\lG_\ell}$ 
does not contain $\fsimp{\Fcal}{i}$,
which is a contradiction.
For each $\ell$ such that $1\leq \ell<s$, since $\lG_\ell\leq\lG_s-1\leq i$,
we have that $z^i_{\lG_\ell}\subseteq G_{\lG_s-1}$,
which means that $\fsimp{\Fcal}{\lG_s-1}\not\in z^i_{\lG_\ell}$
because $\fsimp{\Fcal}{\lG_s-1}\not\in G_{\lG_s-1}$.
Since $\fsimp{\Fcal}{\lG_s-1}\in z^i_{\lG_s}$, 
it follows that $\fsimp{\Fcal}{\lG_s-1}\in\sum_{\ell=1}^s z^i_{\lG_\ell}=\zG$.
This implies that $\zG\nsubseteq G_{\lG_s-1}$
because $\fsimp{\Fcal}{\lG_s-1}\not\in G_{\lG_s-1}$.
However, we have that $j_*\leq\lG_s-1\leq i$ because $j_*<\bar{j}\leq\lG_s\leq i$,
which means that $\zG\subseteq G_{\lG_s-1}$.
So we have reached a contradiction, meaning that $j_*=\bar{j}$.
For each $j\in\upos_{i+1}$,
% unpaired positive index $j$ in the next iteration $i+1$ 
% (i.e., $j\neq j_*$),
if $z^i_j$ does not contain $\fsimp{\Fcal}{i}$,
let $z^{i+1}_j=z^i_j\subseteq G_{i+1}$.
If $z^i_j$ contains $\fsimp{\Fcal}{i}$, 
let $z^{i+1}_j=z^i_j+z^i_{j_*}\subseteq G_{i+1}$.
Note that since $j_*=\bar{j}$,
we must have that $j_*<j$,
which means that $z^i_{j_*}\subseteq G_{k}$ for every $k\in[j,i]\subseteq[j_*,i]$.
Therefore, $z^{i+1}_j=z^i_j+z^i_{j_*}\subseteq G_{k}$ for every $k\in[j,i]$.
Also since $z^i_{j_*}\subseteq G_{j-1}$, 
$z^i_{j_*}$ does not contain $\fsimp{\Fcal}{j-1}$,
which means that $z^{i+1}_j$ contains $\fsimp{\Fcal}{j-1}$.
Note that $\Set{z^{i+1}_j\given j\in\upos_{i+1}}$ must still be linearly independent,
so they form a basis of $\Zyc_1(G_{i+1})$.
Now we have that
$\Set{z^{i+1}_j\given j\in\upos_{i+1}}$ serves as the 1-cycles as specified in the proposition
for iteration $i+1$.
\end{proof}

% Now we can draw our conclusion:

\begin{theorem}\label{thm:1-zigzag-graph-corr}
Algorithm~\ref{alg:1-zigzag-graph-abstr} computes the 1-st zigzag barcode
for a given zigzag filtration on graphs.
\end{theorem}
\begin{proof}
The claim follows directly from Proposition~\ref{prop:pn-paring-w-rep}.
For each interval $[j_*,i]$ produced from the pairing 
in Algorithm~\ref{alg:1-zigzag-graph-abstr},
by the Pairing Principle,
there exists a 1-cycle $z$ containing both 
$\fsimp{\Fcal}{j_*-1}$ and $\fsimp{\Fcal}{i}$
with $z\subseteq G_k$ for every $k\in[j_*,i]$.
The cycle $z$ induces a set of 1-representatives for $[j_*,i]$.
For each interval produced at the end, 
Proposition~\ref{prop:iter-cyc-exist} implies that such an interval
admits 1-representatives.
\end{proof}

\subsection{Efficient implementation}\label{sec:1-zigzag-impl}

For every $i$ and every $j\leq i$,
define
$\Gres_j^i$ as the graph derived from $G_{j}$ by deleting 
every edge $\fsimp{\Fcal}{k}$ 
s.t.\ $j\leq k<i$
and $G_{k}\leftarrowsp{\fsimp{\Fcal}{k}} G_{k+1}$
is backward.
For convenience, we also assume that $\Gres_j^i$
contains all the vertices of $G$.
We can simplify the Pairing Principle
as suggested by the following proposition:
\begin{proposition}
\label{prop:pairing-rewrite}
In each iteration $i$ of Algorithm~\ref{alg:1-zigzag-graph-abstr} 
where $G_{i}\leftarrowsp{\fsimp{\Fcal}{i}} G_{i+1}$
provides $i\in\ninds(\Hm_1(\Fcal))$,
the set $J_i$ in the Pairing Principle can be alternatively
defined as consisting of every $j\in\upos_i$
s.t.\ $\fsimp{\Fcal}{i}\in\Gres_j^{i}$ and
the vertices of $\fsimp{\Fcal}{i}$ are connected in $\Gres_j^{i+1}$ 
{\rm(}$\fsimp{\Fcal}{i}\not\in\Gres_j^{i+1}$ by definition{\rm)}.
\end{proposition}
\begin{proof}%[Proof of Proposition~\ref{prop:pairing-rewrite}]
% {\red (original prop:) In each iteration $i$ of Algorithm~\ref{alg:1-zigzag-graph-abstr} 
% where $G_{i}\leftarrowsp{\fsimp{\Fcal}{i}} G_{i+1}$
% provides $i\in\ninds(\Hm_1(\Fcal))$,
% for an unpaired positive index $j$, 
% let $G'_j$ be the graph derived from $G_j$
% by deleting all edges $\fsimp{\Fcal}{k}$ such that $k\in[j,i-1]$
% and $G_{k}\leftarrowsp{\fsimp{\Fcal}{k}} G_{k+1}$
% is backward.
% Then, the set $J_i$ in the Pairing Principle can be alternatively
% defined as consisting of every unpaired positive index $j$
% such that there is a 1-cycle in $G'_j$ containing $\fsimp{\Fcal}{i}$.}
We prove an equivalent statement, which is that $J_i$
consists of every $j\in\upos_i$
s.t.\ there is a 1-cycle in $\Gres_j^i$ containing $\fsimp{\Fcal}{i}$.
Let $j$ be any index in $\upos_i$.
It is not hard to see that a 1-cycle is in $G_k$ for every $k\in[j,i]$
% if and only if 
iff
the 1-cycle is in $\Gres_j^i$.
So we only need to prove that there is a 1-cycle in $\Gres_j^i$
containing both $\fsimp{\Fcal}{j-1}$ and $\fsimp{\Fcal}{i}$
% if and only if 
iff
there is a 1-cycle in $\Gres_j^i$ containing $\fsimp{\Fcal}{i}$.
The forward direction is easy. So let $z$ be
% there is 
a 1-cycle  in $\Gres_j^i$ containing $\fsimp{\Fcal}{i}$.
If $z$ contains $\fsimp{\Fcal}{j-1}$, then the proof is done.
If not, by Proposition~\ref{prop:iter-cyc-exist} 
% in Section~\ref{sec:1-zigzag-proof},
% presented later,
there is a 1-cycle $z'$ containing $\fsimp{\Fcal}{j-1}$ 
with $z'\subseteq G_k$ for every $k\in[j,i]$.
So $z'$ is a 1-cycle in $\Gres_j^i$ containing $\fsimp{\Fcal}{j-1}$.
If $z'$ contains $\fsimp{\Fcal}{i}$, we again finish our proof.
If not, then $z+z'$ is a 1-cycle 
containing both edges.
% $\fsimp{\Fcal}{j-1}$ and $\fsimp{\Fcal}{i}$.
\end{proof}

% Based on Proposition~\ref{prop:pairing-rewrite},
% {\red for an iteration $i$
% providing a negative index $i$},
% % assume that $\upos_i=\Set{j_1<j_2<\cdots<j_\ell}$.
% % where the indices are ordered increasingly.
% we only need to scan $\upos_i$ from the smallest to the largest index 
% and find the first $j$ s.t.\ the vertices of $\fsimp{\Fcal}{i}$ 
% become connected in $\Gres_j^{i+1}$.
We then turn graphs in $\Fcal$ into weighted ones 
in the following way:
initially, $G_0=\emptyset$;
then, whenever an edge $\fsimp{\Fcal}{i}$ is added from $G_{i}$ to $G_{i+1}$,
the weight $w(\fsimp{\Fcal}{i})$ is set to $i$.
% \sout{Note that edge deletions do not change the weights
% of those edges which are not deleted.}
% while weights of the remaining edges stay the same.
% Note that the weight of an edge can also be interpreted as
% the latest (greatest) index of the addition during which the edge is added.
We have the following fact:

\begin{proposition}\label{prop:G-j-i-charac}
For every $i$ and every $j\leq i$, 
the edge set of $\Gres_j^{i}$, denoted $E(\Gres_j^{i})$,
consists of all edges of $G_{i}$ whose weights are less than $j$.
\end{proposition}
\begin{proof}
We can prove this by induction on $i$. For $i=0$, $G_{i}=\emptyset$
and the proposition is trivially true.
Suppose that the proposition is true for $i$.
If $G_i$ and $G_{i+1}$ differ by a vertex, then the proposition
is also true for $i+1$ because the edges stay the same.
If $G_{i+1}$ is derived from $G_i$ by adding an edge $\fsimp{\Fcal}{i}$,
by the assumption, $E(\Gres_j^{i})$
consists of all edges of $G_{i}$ whose weights are less than $j$
for each $j\leq i$.
Note that $E(\Gres_j^{i})=E(\Gres_j^{i+1})$ because $G_i\rightarrowsp{\fsimp{\Fcal}{i}}G_{i+1}$
is an addition.
So we have that $E(\Gres_j^{i+1})$
consists of all edges of $G_{i+1}$ whose weights are less than $j$
because $w(\fsimp{\Fcal}{i})=i\geq j$.
Since $E(\Gres_{i+1}^{i+1})=E(G_{i+1})$, the claim is also true for $E(\Gres_{i+1}^{i+1})$.
Now consider the situation that 
$G_{i+1}$ is derived from $G_i$ by deleting an edge $\fsimp{\Fcal}{i}$.
% Since $G_i\leftarrowsp{\fsimp{\Fcal}{i}}G_{i+1}$ is an edge deletion,
Then,
$\fsimp{\Fcal}{i}$ must be added to the filtration previously, 
and let $G_k\rightarrowsp{\fsimp{\Fcal}{k}}G_{k+1}$ 
with $k<i$ and $\fsimp{\Fcal}{k}=\fsimp{\Fcal}{i}$ be the {\it latest}
such addition. Note that $w(\fsimp{\Fcal}{i})=k$ in $G_i$.
For $j\leq k$, 
% since $E(\Gres_j^{i})$ consists of all edges of $G_{i}$ 
% whose weights are less than $j$,
$\fsimp{\Fcal}{i}\not\in E(\Gres_j^{i})$ because $w(\fsimp{\Fcal}{i})=k\geq j$.
Since $E(\Gres_j^{i+1})=E(\Gres_j^{i})\setminus\Set{\fsimp{\Fcal}{i}}$,
we have that $E(\Gres_j^{i+1})=E(\Gres_j^{i})$.
Therefore,
$E(\Gres_j^{i+1})$
consists of all edges of $G_{i+1}$ whose weights are less than $j$
because $w(\fsimp{\Fcal}{i})\geq j$ in $G_i$.
For each $j$ s.t.\ $k<j\leq i$,
we have $\fsimp{\Fcal}{i}\in E(\Gres_j^{i})$ and
$E(\Gres_j^{i+1})=E(\Gres_j^{i})\setminus\Set{\fsimp{\Fcal}{i}}$.
Since $w(\fsimp{\Fcal}{i})<j$ in $G_i$,
it is true that $E(\Gres_j^{i+1})$
consists of all edges of $G_{i+1}$ whose weights are less than $j$,
and the proof is done.
\end{proof}

% {\red For an iteration $i$
% providing a negative index $i$},
Suppose that in an iteration $i$ of Algorithm~\ref{alg:1-zigzag-graph-abstr},
$G_{i}\leftarrowsp{\fsimp{\Fcal}{i}} G_{i+1}$
provides a negative index~$i$.
Let $\upos_i=\Set{j_1<j_2<\cdots<j_\ell}$
and the vertices of $\fsimp{\Fcal}{i}$ be $u,v$.
Proposition~\ref{prop:G-j-i-charac} implies that
\begin{equation}\label{eqn:G-j-i-seq}
\Gres_{j_1}^{i+1}\subseteq \Gres_{j_2}^{i+1}\subseteq \cdots 
\subseteq \Gres_{j_\ell}^{i+1}\subseteq \Gres_{i+1}^{i+1}
% {\red \tilde{G}_{i+1}}
\end{equation}
% where $\tilde{G}_{i+1}={G}_{i+1}\union V(G)$ for $V(G)$ being the vertex set of $G$.
% Specifically, for each $\aG$, we have that $\Gres_{j_{\aG+1}}^{i+1}$,
% $\Gres_{j_{\aG}}^{i+1}$ differ by the edges:
% $\Set{e\in E(G_{i+1})\given j_{\aG}\leq w(e)<j_{\aG+1}}$.

By Proposition~\ref{prop:pairing-rewrite},
in order to find
the positive index to pair with $i$,
one only needs to find the smallest $j'_*\in\upos_i$
s.t.\ $u,v$ are connected in $\Gres_{j'_*}^{i+1}$.
(Note that for a $j\in\upos_i$ to satisfy 
$\fsimp{\Fcal}{i}\in\Gres_j^{i}$, 
$j$ only needs to be greater than $w(\fsimp{\Fcal}{i})$;
such a smallest $j$ is easy to derive.)
We further expand Sequence~(\ref{eqn:G-j-i-seq}) into the following finer version:
\begin{equation}\label{eqn:G-j-i-seq-finer}
\Gres_{0}^{i+1}\subseteq \Gres_{1}^{i+1}\subseteq \cdots 
\subseteq \Gres_{i}^{i+1}\subseteq \Gres_{i+1}^{i+1}
\end{equation}
where each consecutive $\Gres_{k}^{i+1},\Gres_{k+1}^{i+1}$
are either the same or differ by only one edge.
To get $j'_*$, we instead scan Sequence~(\ref{eqn:G-j-i-seq-finer})
and find
the smallest $k_*\in\Set{0,1,\ldots,i+1}$
s.t.\ $u,v$ are connected in $\Gres_{k_*}^{i+1}$.
% Then, $j_*$ is the smallest index in $\upos_i$ no less than $k_*$.
% Definition~\ref{dfn:minimax} and 
Proposition~\ref{prop:k-start-in-MST}
characterize such a $k_*$:

% Let $e_0,e_1,\ldots,e_{\innerfiltcnt-1}$ be all the edges in $G_{i+1}$
% s.t.\ $w(e_0)<w(e_1)<\cdots<w(e_{\innerfiltcnt-1})$.
% Then, Sequence~(\ref{eqn:G-j-i-seq}) can be 
% expanded into the following (non-zigzag) filtration:
% \[\innerfilt:\Ginner_0\inctosp{e_0}\Ginner_1\inctosp{e_1}\cdots\inctosp{e_{\innerfiltcnt-1}}\Ginner_\innerfiltcnt\]
% where $V(\Ginner_0)=V(G)$, $E(\Ginner_0)=\emptyset$,
% and $\Ginner_\innerfiltcnt=\tilde{G}_{i+1}$.
% Furthermore, for each $j_\aG\in\upos_{i}$, we have that 
% $\Gres_{j_\aG}^{i+1}=\Ginner_{k}$ for a $k$ s.t.\ $w(e_{k-1})<j_\aG$
% and $w(e_{k})\geq j_\aG$.
% Let the vertices of $\fsimp{\Fcal}{i}$ be $u,v$.
% {\red Recall that}
% the positive index to be paired with $i$ is the smallest index $j_*\in\upos_i$
% s.t.\ $u,v$ are connected in $\Gres_{j_*}^{i+1}$.
% This can be alternatively achieved by finding 
% the smallest index $k_*\in\Set{0,1,\ldots,\innerfiltcnt}$
% s.t.\ $u,v$ are connected in $\Ginner_{k_*}$.
% Then, $j_*$ is the smallest index in $\upos_i$ 

% \begin{definition}[Minimax distance]\label{dfn:minimax}
% For a path in a weighted graph, define the \textbf{max-weight} of the path
% as the maximum weight of its edges.
% Then, for two vertices $x,y$ in the weighted graph,
% the \textbf{minimax distance} of $x,y$ is the minimum max-weight
% of all paths connecting $x$ and $y$.
% \end{definition}

\begin{proposition}\label{prop:k-start-in-MST}
For a path in a weighted graph, let the \textbf{max edge-weight} of the path
be the maximum weight of its edges.
Then,
the integer
$k_*-1$ equals the max edge-weight of the unique path connecting $u,v$
in the unique minimum spanning forest of $G_{i+1}$.
% Furthermore, the minimax distance of $u,v$ in $G_{i+1}$ equals
% the minimax distance of $u,v$
% in the minimum spanning forest of $G_{i+1}$.
\end{proposition}
\begin{proof}
We first notice that, since weighted graphs considered in this paper
have distinct weights, they all have unique minimum spanning forests.
Let $T$ be the minimum spanning forest of $\Gres_{i+1}^{i+1}$;
we prove an equivalent statement of the proposition, 
which is that 
$k_*-1$ equals the max edge-weight of the unique path connecting $u,v$
in $T$.
Since $k_*$ is the smallest index
s.t.\ $u,v$ are connected in $\Gres_{k_*}^{i+1}$,
we have that $u,v$ are {\it not} connected in $\Gres_{k_*-1}^{i+1}$.
This indicates that $\Gres_{k_*}^{i+1}\neq\Gres_{k_*-1}^{i+1}$.
Assume that $\Gres_{k_*}^{i+1}$ and $\Gres_{k_*-1}^{i+1}$ differ
by an edge $e$; then, $w(e)=k_*-1$.
Let $T'$ be the minimum spanning forest of $\Gres_{k_*-1}^{i+1}$.
Then, $u,v$ are not connected in $T'$ and $e$ connects the connected
components of $u,v$ in $T'$. Since spanning forests have a matroid structure,
$T'\union\Set{e}$ must be a subforest of $T$ (indeed, $e$ would be the edge added to $T'$ 
by Kruskal's algorithm; see, e.g.,~\cite{CLRS3rd}).
Also, since there is a unique path between two vertices in a forest,
the path from $u$ to $v$ in $T'\union\Set{e}$
must be the path from $u$ to $v$ in $T$.
This path has a max edge-weight of $k_*-1$ and the proof is done.
% We first prove the first half of the proposition.
% Since $k_*$ is the smallest index
% s.t.\ $u,v$ are connected in $\Gres_{k_*}^{i+1}$,
% we have that $u,v$ are not connected in $\Gres_{k_*-1}^{i+1}$.
% This in turn indicates that
% no path from $u$ to $v$ in $G_{i+1}$ can have a max-weight less than $k_*-1$.
% However, since $u,v$ are connected in $\Gres_{k_*}^{i+1}$,
% there must be a path from $u$ to $v$ with max-weight less than $k_*$,
% which means that the minimax distance of $u,v$ in $G_{i+1}$
% is $k_*-1$.
% We then prove the second half of the proposition.
% Let $F$ be the minimum spanning forest of $G_{i+1}$.
% First, it is obvious that the minimax distance $\dG_1$ of $u,v$ in $G_{i+1}$ 
% is less than or equal to the minimax distance $\dG_2$ of $u,v$ in $F$.
% For contradiction, suppose that $\dG_1<\dG_2$.
% Since $u,v$ are not connected in $\Gres_{k_*-1}^{i+1}$,
% This means that $\Gres_{k_*}^{i+1}$ and $\Gres_{k_*-1}^{i+1}$
% cannot be the same.
% Assume that $\Gres_{k_*}^{i+1}$ and $\Gres_{k_*-1}^{i+1}$ differ
% by an edge $e$; we have that $w(e)=k_*-1$.
% Since $u,v$ are not connected in $\Gres_{k_*-1}^{i+1}$,
% no path of $u$ and $v$ in $G_{i+1}$ can have a max-weight less than $k_*-1$.
\end{proof}

Based on Proposition~\ref{prop:k-start-in-MST},
finding $k_*$ reduces to computing
the max edge-weight of the path
connecting $u,v$ in the minimum spanning forest (MSF) of $G_{i+1}$.
For this, we utilize the {\it dynamic-MSF} data structure 
proposed by Holm et al.~\cite{holm2001poly}.
Assuming that 
% $\filtcnt$ is the length of $\Fcal$ and 
$n$ is the number of vertices and edges of $G$,
the dynamic-MSF data structure supports the following operations:
\begin{itemize}
    \item Return the identifier of a vertex's connected component in $O(\log n)$ time,
    which can be used to determine whether two vertices are connected.
    \item Return the max edge-weight of the path
    connecting any two vertices in the MSF in $O(\log n)$ time.
    \item Insert or delete an edge from the current graph (maintained by the data structure) 
    and possibly update the MSF in $O(\log^4 n)$ amortized time.
\end{itemize}

% We now present the concrete version of the algorithm:
We now present the full details of Algorithm~\ref{alg:1-zigzag-graph-abstr}:
\pagebreak

\begin{algr*}[Algorithm~\ref{alg:1-zigzag-graph-abstr}: details]
\label{alg:1-zigzag-graph-concr}
\begin{itemize}\item[]\end{itemize}\noindent
Maintain a dynamic-MSF data structure $\Fbb$,
which consists of all vertices of $G$ and no edges initially.
Also, set $\upos_0=\emptyset$.
Then, for each $i=0,\ldots,\filtcnt-1$,
if $\fsimp{\Fcal}{i}$ is a vertex, do nothing;
otherwise, do the following:
\vspace{0.3em}
\begin{description}
    \item[Case $G_i\rightarrowsp{\fsimp{\Fcal}{i}}G_{i+1}$:]
    Check 
    whether vertices of $\fsimp{\Fcal}{i}$ are connected in $G_i$
    by querying $\Fbb$, and then add $\fsimp{\Fcal}{i}$ to $\Fbb$.
    If vertices of $\fsimp{\Fcal}{i}$ are connected in $G_i$,
    then 
    % do the following: 
    set $\upos_{i+1}=\upos_i\union\Set{i+1}$;
    otherwise, set $\upos_{i+1}=\upos_i$.
    
    \item[Case $G_i\leftarrowsp{\fsimp{\Fcal}{i}}G_{i+1}$:]
    Delete $\fsimp{\Fcal}{i}$ from $\Fbb$.
    If the vertices $u,v$ of $\fsimp{\Fcal}{i}$ are found to be {\rm not} connected in $G_{i+1}$
    by querying $\Fbb$, 
    then set $\upos_{i+1}=\upos_i$;
    otherwise,
    do the following:
    \begin{itemize}
        \item Find the max edge-weight $w_*$ of the path
        connecting $u,v$ in
        the MSF of $G_{i+1}$ by querying $\Fbb$.
        \item Find the smallest index $j_*$ of $\upos_i$ 
        greater than $\max\Set{w_*,w(\fsimp{\Fcal}{i})}$.
        {\rm(}Note that we can store $\upos_i$ as 
        a red-black tree~\cite{CLRS3rd}, so that finding $j_*$
        takes $O(\log n)$ time.{\rm)}
        \item Output an interval $[j_*,i]$ and set $\upos_{i+1}=\upos_i\setminus\Set{j_*}$.
    \end{itemize}
\end{description}
% \vspace{-0.3em}

At the end, for each $j\in\upos_\filtcnt$,
output an interval $[j,\filtcnt]$.
\end{algr*}

Now we can see that Algorithm~\ref{alg:1-zigzag-graph-abstr} has time complexity $O(\filtcnt\log^4 n)$,
where each iteration is dominated by the update of $\Fbb$.

\section{Codimension-one zigzag persistence of embedded complexes}
\label{sec:alex-dual}

In this section, 
we present an efficient algorithm for computing the $(\Dim-1)$-th
barcode given a zigzag filtration of an $\Real^\Dim$-embedded complex,
by extending our algorithm for 0-dimension with the help of Alexander duality~\cite{munkres2018elements}.

Throughout this section, 
$\Dim\geq 2$, $K$ is a simplicial complex embedded in $\Real^\Dim$,
% ({\red say something about the union thing?}),
and $\Fcal: \emptyset=K_0 \leftrightarrow \cdots \leftrightarrow K_\filtcnt$
is a zigzag filtration of $K$.
We call connected components of $\Real^\Dim\setminus|K|$ as {\it voids} of $K$
or {\it $K$-voids},
% the notation is chosen 
to emphasize that 
only voids of $K$ are considered in this section.
% the latter is for a distinction 
% from those voids of $K$'s subcomplexes in the filtration.
The {\it dual graph} $G$ of $K$ has the vertices corresponding to
the voids as well as the $\Dim$-simplices of $K$,
and has the edges corresponding to the $(\Dim-1)$-simplices of $K$.
The {\it dual filtration} 
$\dfilt: G=G_0 \leftrightarrow G_1 \leftrightarrow \cdots \leftrightarrow G_\filtcnt$ 
of $\Fcal$ consists of subgraphs $G_i$ of $G$ such that: 
% \begin{enumerate}
    % \item
({\sf i})~all vertices of $G$ dual to a $K$-void are in $G_i$;
    % \item
({\sf ii})~a vertex of $G$ dual to a $\Dim$-simplex is in $G_i$ iff
the dual $\Dim$-simplex is {\it not} in $K_i$;
    % \item
({\sf iii})~an edge of $G$ is in $G_i$ iff
its dual $(\Dim-1)$-simplex is {\it not} in $K_i$.
% \end{enumerate}

One could verify that each $G_i$ is a well-defined subgraph of $G$.
We note the following: ({\sf i})~inclusion directions in $\dfilt$ are reversed;
({\sf ii})~$\dfilt$ is not exactly a zigzag filtration (because
an arrow may introduce no changes)
but can be easily made into one.
Figure~\ref{fig:df} gives an example
% of a zigzag filtration 
in $\Real^2$,
% and its dual filtration,
in which we observe the following: whenever a $(\Dim-1)$-cycle 
(i.e., 1-cycle) is formed in the primal filtration,
a connected component in the dual filtration splits;
whenever a $(\Dim-1)$-cycle 
is killed in the primal filtration,
a connected component in the dual filtration vanishes.
Intuitively, $G_i$ encodes the connectivity 
of $\Real^\Dim\setminus|K_i|$,
and so by Alexander duality,
% and its naturality,
% the barcodes of $\Hm_{\Dim-1}(\Fcal)$ and $\Hmr_0(\dfilt)$
% are equivalent.
% So 
we have the following proposition: 
% {\red(proof in Appendix~\ref{sec:pf-prop-equiv-barc})}:

\begin{figure}
  \centering
  \includegraphics[width=0.75\linewidth]{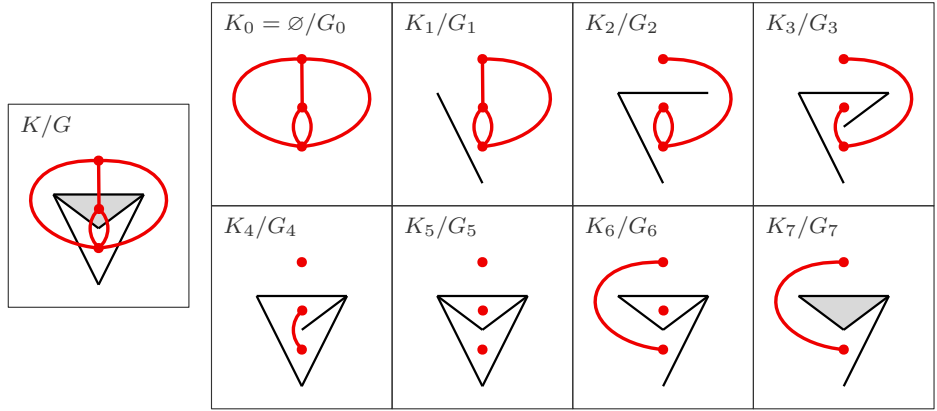}
  \caption{A zigzag filtration of an $\Real^2$-embedded complex $K$
  and its dual filtration, where the dual graphs are colored red.
  For brevity, changes of vertices in the primal filtration are ignored.}
  \label{fig:df}
\end{figure}

\begin{proposition}\label{prop:equiv-barc}
% Let $\Fcal$ and $\dfilt$ be as defined.
% One has that 
$\Pers(\Hm_{\Dim-1}(\Fcal))=\Pers(\Hmr_0(\dfilt))$.
\end{proposition}
% \begin{proof}
% See Appendix~\ref{sec:pf-thm-equiv-barc}.
% \end{proof}

We conduct the proof of Proposition~\ref{prop:equiv-barc}
by introducing two modules $\Mcal_3$, $\Mcal_4$ and proving the following:
\begin{itemize}
\item
$\Pers(\Hm_{\Dim-1}(\Fcal))=\Pers(\Mcal_3)$, in Proposition~\ref{prop:equiv-barc-M3}.
\item
$\Pers(\Mcal_3)=\Pers(\Mcal_4)$, in Proposition~\ref{prop:equiv-barc-M3-M4}.
\item
$\Pers(\Mcal_4)=\Pers(\Hmr_0(\dfilt))$, in Proposition~\ref{prop:equiv-barc-M4}.
\end{itemize}

\begin{proposition}\label{prop:equiv-barc-M3}
The following zigzag modules are isomorphic:
% \begin{align*}
\[
% \centerline
{\xymatrix{
\Hm_{\Dim-1}(\Fcal):\hspace{-5em}
  & \Hm_{\Dim-1}(K_0) \ar[d]^(.43){\rotatebox{90}{$\approx$}} \ar@{<->}[r]^(.52){\iG_*} 
  & \Hm_{\Dim-1}(K_1) \ar[d]^(.43){\rotatebox{90}{$\approx$}} \ar@{<->}[r]^(.72){\iG_*} 
  & \cdots \ar@{<->}[r]^(.29){\iG_*} 
  & \Hm_{\Dim-1}(K_\filtcnt) \ar[d]^(.43){\rotatebox{90}{$\approx$}}\\ 
\Mcal_2:\hspace{-2.8em}
  & \Hmr^{0}(\Real^\Dim\setminus|K_0|) \ar[d]^(.43){\rotatebox{90}{$\approx$}} \ar@{<->}[r]^(.52){\iG^*} 
  & \Hmr^{0}(\Real^\Dim\setminus|K_1|) \ar[d]^(.43){\rotatebox{90}{$\approx$}} \ar@{<->}[r]^(.72){\iG^*} 
  & \cdots  \ar@{<->}[r]^(.29){\iG^*} 
  & \Hmr^{0}(\Real^\Dim\setminus|K_\filtcnt|) \ar[d]^(.43){\rotatebox{90}{$\approx$}} \\ 
\Mcal_3:\hspace{-2.8em}
  & \Hom(\Hmr_{0}(\Real^\Dim\setminus|K_0|)) \ar@{<->}[r]^(.52){(\iG_*)^*} 
  & \Hom(\Hmr_{0}(\Real^\Dim\setminus|K_1|)) \ar@{<->}[r]^(.75){(\iG_*)^*} 
  & \cdots  \ar@{<->}[r]^(.29){(\iG_*)^*} 
  & \Hom(\Hmr_{0}(\Real^\Dim\setminus|K_\filtcnt|)) \\ 
}}
\]
% \end{align*}
which implies that $\Pers(\Hm_{\Dim-1}(\Fcal))=\Pers(\Mcal_3)$.
\end{proposition}
\begin{proof}
In the diagram, $\iG_*$ and $\iG^*$ are induced by inclusion,
$\Hom(\Hmr_{0}(\Real^\Dim\setminus|K_i|))$ is the space of all linear maps
from $\Hmr_{0}(\Real^\Dim\setminus|K_i|)$ to $\Zbb_2$,
and $(\iG_*)^*$ is the dual of $\iG_*$ (i.e., $(\iG_*)^*(g)=g\circ\iG_*$).
The isomorphism between $\Hm_{\Dim-1}(\Fcal)$ and $\Mcal_2$
is given by Alexander duality~\cite[Corollary~72.4]{munkres2018elements},
and the isomorphism between $\Mcal_2$ and $\Mcal_3$ is
given by the universal coefficient theorem~\cite[pg.~196,198]{hatcher2002algebraic}.
\end{proof}

\begin{proposition}\label{prop:equiv-barc-M3-M4}
Define an elementary zigzag module $\Mcal_4$ as:
\[
\Mcal_4:
\Hmr_{0}(\Real^\Dim\setminus|K_0|) \leftrightarrowsp{\iG_*}
\Hmr_{0}(\Real^\Dim\setminus|K_1|) \leftrightarrowsp{\iG_*}
\cdots  \leftrightarrowsp{\iG_*}
\Hmr_{0}(\Real^\Dim\setminus|K_\filtcnt|)
\]
Then, one has that $\Pers(\Mcal_3)=\Pers(\Mcal_4)$.
\end{proposition}
\begin{proof}
First note that directions of corresponding arrows in $\Mcal_3$ and $\Mcal_4$
are reversed.
To prove the proposition, 
first observe that for each linear map
% \footnote{
%   \red\it In this paper, I tend to use the term `linear maps' rather than `morphisms'.}
$\iG_*$ in $\Mcal_4$ and 
its corresponding map $(\iG_*)^*$ in $\Mcal_3$,
if $\iG_*$ is an isomorphism, then $(\iG_*)^*$ is an isomorphism;
if $\iG_*$ is injective with non-trivial cokernel, 
then $(\iG_*)^*$ is surjective with non-trivial kernel;
if $\iG_*$ is surjective with non-trivial kernel, 
then $(\iG_*)^*$ is injective with non-trivial cokernel.
Therefore, $\pinds(\Mcal_3)=\pinds(\Mcal_4)$ and $\ninds(\Mcal_3)=\ninds(\Mcal_4)$.
So intervals in $\Pers(\Mcal_4)$ also induce a bijection 
from $\pinds(\Mcal_3)$ to $\ninds(\Mcal_3)$,
by mapping the start of each interval to the end.
Based on Proposition~\ref{prop:pn-paring-w-rep}, 
in order to prove that $\Pers(\Mcal_3)=\Pers(\Mcal_4)$,
we only need to show that every interval in $\Pers(\Mcal_4)$
admits a set of representatives in the module $\Mcal_3$.
Let $\Mcal_4=\bigoplus_{k\in\LG}\Ical^{[b_k,d_k]}$ 
be an interval decomposition of $\Mcal_4$.
Moreover,
for each $k\in\LG$ and each $i\in[b_k,d_k]$, 
let $\Ical^{[b_k,d_k]}(i)$ be the $i$-th vector space in $\Ical^{[b_k,d_k]}$
and let $\aG^k_i$ be the non-zero element 
in $\Ical^{[b_k,d_k]}(i)$.
Then, for each $k\in\LG$, define a set of representatives 
$\bigSet{\bG_i\in\Hom(\Hmr_{0}(\Real^\Dim\setminus|K_i|))
\given i\in[b_k,d_k]}$ for $[b_k,d_k]$ in $\Mcal_3$
% {\red interval submodule} $\Jcal^{[b_k,d_k]}$ of $\Mcal_3$
as follows:
for each $i\in[b_k,d_k]$, let 
$\Bcal=\bigSet{\aG^\ell_i\given \ell\in\LG\text{ and }[b_\ell,d_\ell]\ni i}$ 
and note that $\Bcal$ forms a basis of $\Hmr_{0}(\Real^\Dim\setminus|K_i|)$;
then, $\bG_i$ maps 
$\aG^k_i$ to 1 and all the other elements in $\Bcal$ to 0.
It can then be verified that $\Set{\bG_i\given i\in[b_k,d_k]}$
forms a valid set of representatives for $[b_k,d_k]$ in $\Mcal_3$.
\end{proof}

\begin{proposition}\label{prop:equiv-barc-M4}
The following zigzag modules are isomorphic:
\[
% \centerline
{\xymatrix{
\Hmr_0(\dfilt):\hspace{-3.8em}
  & \Hmr_{0}(G_0) \ar[d]_(.44)\thG^(.43){\rotatebox{90}{$\approx$}} \ar@{<->}[r]^(.5){\iG_*} 
  & \Hmr_{0}(G_1) \ar[d]_(.44)\thG^(.43){\rotatebox{90}{$\approx$}} \ar@{<->}[r]^(.68){\iG_*} 
  & \cdots  \ar@{<->}[r]^(.33){\iG_*} 
  & \Hmr_{0}(G_\filtcnt) \ar[d]_(.43)\thG^(.44){\rotatebox{90}{$\approx$}} \\ 
\Mcal_4:\hspace{-2.8em}
  & \Hmr_{0}(\Real^\Dim\setminus|K_0|) \ar@{<->}[r]^(.5){\iG_*} 
  & \Hmr_{0}(\Real^\Dim\setminus|K_1|) \ar@{<->}[r]^(.68){\iG_*} 
  & \cdots  \ar@{<->}[r]^(.33){\iG_*} 
  & \Hmr_{0}(\Real^\Dim\setminus|K_\filtcnt|) \\ 
}}
\]
which implies that $\Pers(\Mcal_4)=\Pers(\Hmr_0(\dfilt))$.
\end{proposition}
\begin{proof}
In the diagram, 
$\thG$ is induced by a map $\tilde{\thG}$ which 
maps a vertex of $G_i$ to a point in the corresponding $K$-void or $\Dim$-simplex
contained in $\Real^\Dim\setminus|K_i|$,
so the commutativity of the diagram can be easily verified.
To see that $\thG$ is an isomorphism, we observe the following facts, 
which can be verified by induction:
\begin{itemize}
    % \item $\tilde{\thG}$ encodes 
    % a bijection from vertices of $G_i$
    % to $K$-voids and $\Dim$-simplices contained in $\Real^\Dim\setminus|K_i|$.
    \item Any point in $\Real^\Dim\setminus|K_i|$ is path connected to a point
    in a $K$-void or a $\Dim$-simplex contained in $\Real^\Dim\setminus|K_i|$.
    \item For any two vertices $v_1,v_2$ of $G_i$, $v_1$ is path connected to $v_2$
    in $G_i$ if and only if $\tilde{\thG}(v_1)$ is path connected to 
    $\tilde{\thG}(v_2)$ in $\Real^\Dim\setminus|K_i|$.
    \qedhere
\end{itemize}
\end{proof}
% \xymatrixcolsep{4pc}\xymatrixrowsep{8pc}
%      \xymatrix{
%          & 
%          \txt{N} \ar@<1.5ex>[dr]^{q{N}{V}} \ar@<0.5ex>[dl]^{q{N}{U}} 
%          &  
%          \\
%          \txt{U} \ar@<1.5ex>[ur]^{q{U}{N}} \ar@<0.5ex>[rr]^{q{U}{V}} 
%          &
%          & 
%          \txt{V} \ar@<0.5ex>@*{[|(3)]}[lu]^{q{V}{N}}  \ar@<1.5ex>[ll]^{q{V}{U}}
%       }

% By Theorem~\ref{prop:equiv-barc},
% in order to compute $\Pers(\Hm_{\Dim-1}(\Fcal))$, 
% one only needs to compute $\Pers(\Hmr_0(\dfilt))$.
Proposition~\ref{prop:reduced-barcode} 
% {\red(proof in Appendix~\ref{sec:pf-prop-reduced-barcode})}
indicates a way to compute $\Pers(\Hmr_0(\dfilt))$:
% utilizing Algorithm~\ref{alg:0-zigzag}:

\begin{proposition}\label{prop:reduced-barcode}
$\Pers(\Hmr_0(\dfilt))=\Pers(\Hm_0(\dfilt))\setminus\Set{[0,m]}$.
\end{proposition}
% \begin{proof}
% See Appendix~\ref{sec:pf-prop-reduced-barcode}.
% \end{proof}
\begin{proof}
We first note that $\Hm_0(\dfilt)$ is not an elementary module
because it does not start with the trivial vector space,
whereas most of our definitions and conclusions assume elementary ones.
However, we can add a trivial vector space at the beginning of $\Hm_0(\dfilt)$
to make it elementary so that
those constructions remain valid.
We then show that $[0,\filtcnt]\in\Pers(\Hm_0(\dfilt))$ by simulating
Algorithm~\ref{alg:0-zigzag} on $\dfilt$.
However, $\dfilt$ does not start with an empty complex
as assumed by Algorithm~\ref{alg:0-zigzag}.
Hence, the lowest node of $\bcgraph(\dfilt)$ is at level $0$
instead of $1$; the latter is the case for a filtration
starting with an empty complex.
%{\red However, Algorithm~\ref{alg:0-zigzag} can be easily adapted to work on $\bcgraph(\dfilt)$.}
With this minor adjustment, Algorithm~\ref{alg:0-zigzag}
can still be applied on $\bcgraph(\dfilt)$.
Let $\bcforest_i$ be the barcode forest in iteration $i$ of
Algorithm~\ref{alg:0-zigzag} when taking $\dfilt$ as input.
Then, it can be proved by induction that for every $i$, 
$\bcforest_i$ contains
a tree $\UG_i$ with a level-0 root and all the connected components 
of $G_i$, which contain vertices dual to the $K$-voids,
correspond to leaves in $\UG_i$.
Therefore, Algorithm~\ref{alg:0-zigzag} must produce an interval $[0,\filtcnt]$
at the end.

Note that $\pinds(\Hmr_0({\Ecal}))=\pinds(\Hm_0({\Ecal}))\setminus\Set{0}$
and $\ninds(\Hmr_0(\dfilt))=\ninds(\Hm_0(\dfilt))\setminus\Set{\filtcnt}$,
so $\Pers(\Hm_0(\dfilt))\setminus\Set{[0,\filtcnt]}$ induces
a bijection from $\pinds(\Hmr_0(\dfilt))$
to $\ninds(\Hmr_0(\dfilt))$.
By Proposition~\ref{prop:pn-paring-w-rep}, we only need to show that
every interval in $\Pers(\Hm_0(\dfilt))\setminus\Set{[0,m]}$
admits representatives in the module $\Pers(\Hmr_0(\dfilt))$.
Let $\Set{\aG_i\given i\in[s,t]}$ be an arbitrary set of representatives for 
an interval $[s,t]\subseteq[0,\filtcnt]$
in the module $\Pers(\Hm_0(\dfilt))$; we note the following facts:
\begin{itemize}
    \item The parity of the number of connected components in
each $\aG_i$ for $i\in[s,t]$ is the same.
     \item A 0-th homology class is in the reduced homology group
if it consists of an even number of connected components.
\end{itemize}
For an interval $[b,d]$ in $\Pers(\Hm_0(\dfilt))\setminus\Set{[0,m]}$,
let $\Set{\bG_i\given i\in[b,d]}$ be a set of representatives for $[b,d]$
in the module $\Pers(\Hm_0(\dfilt))$.
If $\psi^0_{b-1}$ is backward, where $\psi^0_{b-1}$ is the linear map 
in $\Hm_0(\dfilt)$ connecting $\Hm_0(G_{b-1})$ and $\Hm_0(G_{b})$,
then $\psi^0_{b-1}(\bG_b)=0$ by definition.
It follows that $\bG_b$ consists of 
an even number of connected components.
Therefore, 
% Because a sum of an even number of connected components always
% belongs to the 0-th reduced homology groups, we have that 
$\Set{\bG_i\given i\in[b,d]}$ is also a set of representatives for $[b,d]$
in the module $\Pers(\Hmr_0(\dfilt))$ by the noted facts.
Similar situations happen when $d<\filtcnt$ and $\psi^0_{d}$ is forward.
Now suppose that $\psi^0_{b-1}$ is forward 
and either $d=\filtcnt$ or $\psi^0_{d}$ is backward.
If $\bG_b$ consists of an even number of connected components,
% then $\bG_i$ consists of 
% an even number of connected components for each $i\in[b,d]$,
then $\Set{\bG_i\given i\in[b,d]}$ is also a set of representatives for $[b,d]$
in the module $\Pers(\Hmr_0(\dfilt))$.
If $\bG_b$ consists of an odd number of connected components,
let $\Set{\aG_i\given i\in[0,\filtcnt]}$ be
a set of representatives for the interval $[0,\filtcnt]$
in the module $\Pers(\Hm_0(\dfilt))$.
Then,
$\Set{\bG_i+\aG_i\given i\in[b,d]}$ is a set of representatives for $[b,d]$
in the module $\Pers(\Hmr_0(\dfilt))$
because $\aG_0$ consists of a single connected component.
\end{proof}

% \subsection{Algorithm}\label{sec:alg-codim-1}

The above two propositions
% Proposition~\ref{prop:equiv-barc} and~\ref{prop:reduced-barcode} 
% Section~\ref{sec:df-bceq}
suggest a naive algorithm for computing 
$\Pers(\Hm_{\Dim-1}(\Fcal))$
% the $(\Dim-1)$-th barcode of $\Fcal$
using Algorithm~\ref{alg:0-zigzag}. 
However, 
% in order to 
building the dual graph $G$ 
% for $K$, 
% one has to retrieve the boundary of each void of $K$.
requires reconstructing the void boundaries of $K$,
% a common technique is to apply
which is done by
a ``walking'' algorithm obtaining a set of $(\Dim-1)$-cycles
and then by a nesting test of these $(\Dim-1)$-cycles~\cite[Section 4.1]{dey2020computing}.
% and Figure~\ref{fig:void_bd1}.
The running time of this process is $\OG(n^2)$ where $n$
is the size of $K$. 
% We will describe in the following 
% how we avoid the time-consuming nesting test and achieve a sub-quadratic time complexity.
To achieve the claimed complexity,
we first define the following~\cite{dey2020computing}:

\begin{definition}\label{dfn:q-conn}
In a simplicial complex $X$,
two $\diml$-simplices $\sG,\sG'$
% ($\diml\geq 1$)
are \textbf{$\diml$-connected}
if there is a sequence $\sG_0,\ldots,\sG_\ell$ of $\diml$-simplices of $X$
such that $\sG_0=\sG$, $\sG_\ell=\sG'$, 
and every $\sG_i,\sG_{i+1}$ share a $(\diml-1)$-face.
% for each $0\leq i<\ell$.
A maximal set of $\diml$-connected $\diml$-simplices of $X$ 
is called a \textbf{$\diml$-connected component} of $X$,
and $X$ is \textbf{$\diml$-connected} 
if it has only one $\diml$-connected component.
\end{definition}

Based on the fact that void boundaries
can be reconstructed without the nesting test 
for $(\Dim-1)$-connected complexes in $\Real^\Dim$~\cite{dey2020computing},
% (see~\cite{dey2020computing} and Figure~\ref{fig:void_bd2}),
we restrict $\Fcal$ to several $(\Dim-1)$-connected subcomplexes of $K$
% The $(\Dim-1)$-th barcode of $\Fcal$ 
% $\Pers(\Hm_{\Dim-1}(\Fcal))$
and then take the union of
the $(\Dim-1)$-th barcodes of these restricted filtrations:

\begin{algr}[Algorithm for $(\Dim-1)$-th zigzag persistence on $\Real^\Dim$-embedded complexes]
\label{alg:codim-1-zigzag}
% \begin{itemize}\item[]\end{itemize}\noindent
\hspace{3em}
%\sout{Given a simplicial complex $K$ embedded in $\Real^\Dim$ and a zigzag filtration 
%$\Fcal: \emptyset=K_0 \leftrightarrow \cdots \leftrightarrow K_\filtcnt$ 
%of $K$, the algorithm does the following:}
\begin{enumerate}
    \item Compute the $(\Dim-1)$-connected components $\comp{1},\ldots,\comp{\compcnt}$
    of $K$.
    \item For each $\ell=1,\ldots,\compcnt$, let 
    \[\compcmplx{\ell}=\cl{\comp{\ell}}\union\Set{\tG\in K\given 
    \tG\text{ is a }\Dim\text{-simplex whose }(\Dim-1)\text{-faces are in }\comp{\ell}}\]
    and let $\rfilt{\ell}$ be a filtration of $\compcmplx{\ell}$ defined as
    \[\rfilt{\ell}:K_0\intersect\compcmplx{\ell} \leftrightarrow 
    K_1\intersect\compcmplx{\ell} \leftrightarrow \cdots \leftrightarrow K_\filtcnt\intersect\compcmplx{\ell}\]
    Then, compute $\Pers(\Hm_{\Dim-1}(\rfilt{\ell}))$ for each $\ell$.
    % using the method described above.
    \item Return $\bigunion_{\ell=1}^\compcnt\Pers(\Hm_{\Dim-1}(\rfilt{\ell}))$
    as the $(\Dim-1)$-th barcode of $\Fcal$.
\end{enumerate}
\end{algr}

In Algorithm~\ref{alg:codim-1-zigzag},
% Note that 
$\cl{\comp{\ell}}$ denotes the closure of $\comp{\ell}$, i.e., 
the complex consisting of all faces of simplices in $\comp{\ell}$.
% Also note that $\rfilt{\ell}$
% in Algorithm~\ref{alg:codim-1-zigzag}
% is not exactly a zigzag filtration (because
% other than adding or deleting a simplex,
% an arrow may introduce no changes)
% to the subsequent complex.
% but can be easily made into one.
% by ignoring those arrows with no changes.
For each $\ell$,
let $n_\ell$ be the number of simplices in $\compcmplx{\ell}$;
then, the dual graph of $\compcmplx{\ell}$ 
can be constructed in $O(n_\ell\log n_\ell)$
time because $\compcmplx{\ell}$ is $(\Dim-1)$-connected~\cite{dey2020computing}.
Using Algorithm~\ref{alg:0-zigzag}
% as suggested by
% Theorem~\ref{prop:equiv-barc} and Proposition~\ref{prop:reduced-barcode} 
% described in Section~\ref{sec:df-bceq}
to compute $\Pers(\Hm_{\Dim-1}(\rfilt{\ell}))$
as suggested by the duality,
the running time of Algorithm~\ref{alg:codim-1-zigzag}
is $O(m\log^2 n+m\log m+n\log n)$,
where $m$ is the length of $\Fcal$ and $n$ is the size of $K$.

The correctness of Algorithm~\ref{alg:codim-1-zigzag} can be seen from the following proposition:

\begin{proposition}\label{prop:codim-1-mod-sum}
The modules $\Hm_{\Dim-1}(\Fcal)$ and $\bigoplus_{\ell=1}^\compcnt\Hm_{\Dim-1}(\rfilt{\ell})$ 
are isomorphic
% \[
% \centerline
% {\xymatrix{
% \Hm_{\Dim-1}(\Fcal):\hspace{0.7em}
%   & \cdots \ar@{<->}[r]^(.4){\morph{\Fcal}{i-1}{\Dim-1}} 
%   & \Hm_{\Dim-1}(K_i) \ar[d]^(.43){\rotatebox{90}{$\approx$}} \ar@{<->}[r]^(.5){\morph{\Fcal}{i}{\Dim-1}} 
%   & \Hm_{\Dim-1}(K_{i+1}) \ar[d]^(.43){\rotatebox{90}{$\approx$}} 
%   \ar@{<->}[r]^(.62){\morph{\Fcal}{i+1}{\Dim-1}} 
%   & \cdots  
% %   \ar@{<->}[r]^(.33){\iG_*} 
% %   & \Hm_{\Dim-1}(K_\filtcnt) \ar[d]_(.43)\thG^(.44){\rotatebox{90}{$\approx$}} 
%   \\ 
% \bigoplus_{\ell=1}^\compcnt\Hm_{\Dim-1}(\rfilt{\ell}):\hspace{-2em}
%   & \cdots \ar@{<->}[r]%^(.27){\psi_{i-1}} 
%   & \bigoplus_{\ell=1}^\compcnt\Hm_{\Dim-1}(K_i\intsec\compcmplx{\ell}) \ar@{<->}[r]%^(.485){\psi_{i}} 
%   & \bigoplus_{\ell=1}^\compcnt\Hm_{\Dim-1}(K_{i+1}\intsec\compcmplx{\ell}) 
%   \ar@{<->}[r]%^(.76){\psi_{i+1}} 
%   & \cdots  
% %   \ar@{<->}[r]^(.33){\iG_*} 
% %   & \bigoplus_{\ell=1}^\compcnt\Hm_{\Dim-1}(K_0\intsec\compcmplx{\ell}) 
%   \\ 
% }}
% \]
% In the above diagram, 
% $\psi_{i}$ is the direct sum of each linear map induced by inclusion
% which connects $\Hm_{\Dim-1}(K_i\intsec\compcmplx{\ell})$ 
% and $\Hm_{\Dim-1}(K_{i+1}\intsec\compcmplx{\ell})$.
which implies that 
\[\Pers(\Hm_{\Dim-1}(\Fcal))=\bigunion_{\ell=1}^\compcnt\Pers(\Hm_{\Dim-1}(\rfilt{\ell}))\]
\end{proposition}
% \begin{proof}
% Intuitively, $\compcmplx{1},\ldots,\compcmplx{\compcnt}$ are derived 
% from $(\Dim-1)$-connected components of $K$,
% so their $(\Dim-1)$-th homology is independent.
% See Appendix~\ref{sec:pf-codim-1-mod-sum} for a formal justification.
% \end{proof}

To prove Proposition~\ref{prop:codim-1-mod-sum},
we define the following notations:
for a simplicial complex $X$, $\simpset_{\diml}(X)$ denotes the set of $\diml$-simplices
of $X$; also, for a $\diml$-chain $A$ of $X$ 
and a $k$-simplex $\sG$ of $X$ with $k<\diml$, 
$\cof_{\diml}(A,\sG)$ denotes the set of $\diml$-cofaces of $\sG$ belonging to $A$.

We first prove the following proposition:

\begin{proposition}\label{prop:dirsum-homo}
Let $\diml\geq2$ and 
$X$ be a simplicial complex with subcomplexes $X_1$, $X_2$, and $X'$.
If $\simpset_{\diml-1}(X_1)\union\simpset_{\diml-1}(X_2)=\simpset_{\diml-1}(X)$,
$\simpset_{\diml}(X_1)\union\simpset_{\diml}(X_2)=\simpset_{\diml}(X)$,
and no $(\diml-1)$-simplex of $X_1$ is $(\diml-1)$-connected to a 
$(\diml-1)$-simplex of $X_2$, 
then the following diagram commutes:
\[
% \centerline
{\xymatrix{
\Hm_{\diml-1}(X_1\intsec X')\oplus\Hm_{\diml-1}(X_2\intsec X') 
\ar[d]_(.44){\thG_1}^(.43){\rotatebox{90}{$\approx$}} \ar@{<->}[r]^(.585){\iG_*} 
  & \Hm_{\diml-1}(X_1)\oplus\Hm_{\diml-1}(X_2) 
  \ar[d]_(.44){\thG_2}^(.43){\rotatebox{90}{$\approx$}} \\
\Hm_{\diml-1}(X') \ar@{<->}[r]^(.585){\iG_*} 
  & \Hm_{\diml-1}(X) \\
}}
\]
where $\thG_1$ and $\thG_2$ are isomorphisms, the upper $\iG_*$ is the direct sum
of the linear maps induced by inclusions, and the lower $\iG_*$ is induced by 
the inclusion.
\end{proposition}
\begin{remark}\label{rmk:ext2gen}
The above proposition can be easily extended to the general case with 
a finite number of subcomplexes $X_1,X_2,\ldots,X_k$ of $X$ such that 
$\bigunion_{\ell=1}^k\simpset_{\diml-1}(X_\ell)=\simpset_{\diml-1}(X)$,
$\bigunion_{\ell=1}^k\simpset_{\diml}(X_\ell)=\simpset_{\diml}(X)$,
and $(\diml-1)$-simplices of each pair of $X_1,\ldots,X_k$ are not $(\diml-1)$-connected.
\end{remark}
\begin{proof}
In the diagram, $\thG_2$ is defined as follows ($\thG_1$ is similarly defined): 
for each $\big([z_1],[z_2]\big)\in\Hm_{\diml-1}(X_1)\oplus\Hm_{\diml-1}(X_2)$,
$\thG_2\big([z_1],[z_2]\big)=[z_1+z_2]$. 
It is not hard to see that
the diagram is commutative and $\thG_1$, $\thG_2$ are linear maps.
Therefore, we only need to show that $\thG_1$, $\thG_2$ are isomorphisms.
We only do this for $\thG_2$ because the verification for $\thG_1$ is similar.

For the surjectivity of $\thG_2$, let $[z]\in\Hm_{\diml-1}(X)$ be arbitrary and
let $z_1=z\intsec X_1$, $z_2=z\intsec X_2$.
We claim that $\partial(z_1)=\partial(z_2)=0$, which implies that
$\thG_2\big([z_1],[z_2]\big)=[z]$ and hence the surjectivity.
We only show that $\partial(z_1)=0$ because for $\partial(z_2)$
it is similar. For contradiction, suppose instead that there is a $(\diml-2)$-simplex
$\sG\in\partial(z_1)$. Then, $|\cof_{\diml-1}(z_1,\sG)|$ is an odd number.
Note that $|\cof_{\diml-1}(z,\sG)|$ is an even number 
and $\cof_{\diml-1}(z_1,\sG)\subseteq\cof_{\diml-1}(z,\sG)$,
so $\cof_{\diml-1}(z,\sG)\setminus\cof_{\diml-1}(z_1,\sG)$ is not empty.
Let $\tG\in\cof_{\diml-1}(z,\sG)\setminus\cof_{\diml-1}(z_1,\sG)$;
then $\tG\not\in X_1$ because if $\tG\in X_1$,
$\tG$ must belong to $z_1$ and hence belongs to $\cof_{\diml-1}(z_1,\sG)$. 
So $\tG\in X_2$.
Note that $\tG$ must be $(\diml-1)$-connected to a $\tG'\in\cof_{\diml-1}(z_1,\sG)$
which is in $X_1$ because they share a common $(\diml-2)$-face $\sG$.
But this contradicts the fact that 
no $(\diml-1)$-simplex of $X_1$ is $(\diml-1)$-connected to a 
$(\diml-1)$-simplex of $X_2$,
and so $\partial(z_1)=0$.

For the injectivity of $\thG_2$, let $\big([z_1],[z_2]\big)$ be
any element of $\Hm_{\diml-1}(X_1)\oplus\Hm_{\diml-1}(X_2)$
such that $\thG_2\big([z_1],[z_2]\big)=[z_1+z_2]=0$. 
Because $z_1+z_2$ is a boundary in $X$, let $A$ be the $\diml$-chain
in $X$ such that $\partial(A)=z_1+z_2$. Moreover, 
let $A_1=A\intsec X_1$ and $A_2=A\intsec X_2$.
We claim that $z_1=\partial(A_1)$ and $z_2=\partial(A_2)$,
which implies that $\big([z_1],[z_2]\big)=(0,0)$ and hence the injectivity.
We only show that $z_1=\partial(A_1)$ because the other is similar:
\begin{description}
    \item[$z_1\subseteq\partial(A_1)$:]
For any $(\diml-1)$-simplex $\sG\in z_1$, because $\sG\in X_1$ 
and no $(\diml-1)$-simplex of $X_1$ is $(\diml-1)$-connected to a 
$(\diml-1)$-simplex of $X_2$, $\sG\not\in X_2$.
So $\sG\not\in z_2$, which means that $\sG\in X_1+X_2=\partial(A)$.
We claim that $\cof_\diml(A,\sG)=\cof_\diml(A_1,\sG)$.
Then,
the fact that $|\cof_\diml(A,\sG)|$ is an odd number implies that
$|\cof_\diml(A_1,\sG)|$ is an odd number, and hence $\sG\in\partial(A_1)$.
To prove that $\cof_\diml(A,\sG)=\cof_\diml(A_1,\sG)$,
first note that $\cof_\diml(A_1,\sG)\subseteq\cof_\diml(A,\sG)$.
We then show that $\cof_\diml(A,\sG)\subseteq\cof_\diml(A_1,\sG)$.
Let $\tG$ be any $\diml$-simplex in $\cof_\diml(A,\sG)$.
% The fact that $\sG\in X_1$ implies that any $(\diml-1)$-simplex in $\partial(\tG)$
% must be in $X_1$ because no $(\diml-1)$-simplex of $X_1$ is $(\diml-1)$-connected to a 
% $(\diml-1)$-simplex of $X_2$.
We must have that $\tG\in X_1$ because otherwise $\tG\in X_2$
which implies that $\sG\in X_2$, a contradiction.
It follows that $\tG\in A_1=A\intsec X_1$ 
and hence $\tG\in\cof_\diml(A_1,\sG)$.
    \item[$\partial(A_1)\subseteq z_1$:]
    The proof is similar to the previous one and is omitted.
% For any $\sG\in\partial(A_1)$, we have that $\sG\in X_1$ because $A_1\subseteq X_1$.
% We claim that $\cof_\diml(A,\sG)=\cof_\diml(A_1,\sG)$. 
% Then,
% the fact that $|\cof_\diml(A_1,\sG)|$ is an odd number, it follows that
% $|\cof_\diml(A,\sG)|$ is also an odd number 
% and hence $\sG\in\partial(A)=z_1+z_2$.
% Now $\sG$ must come from $z_1$ or $z_2$,
% and because $\sG\in X_1$, we have that $\sG\in z_1$.
% The proof of $\cof_\diml(A,\sG)=\cof_\diml(A_1,\sG)$
% is similar to the previous one and is omitted.
\qedhere
\end{description}
\end{proof}

Now we prove Proposition~\ref{prop:codim-1-mod-sum}.
First note that for each $i$, 
$\bigunion_{\ell=1}^\compcnt\simpset_{\Dim-1}(K_{i+1}\intsec\compcmplx{\ell})
=\simpset_{\Dim-1}(K_{i+1})$
and $\bigunion_{\ell=1}^\compcnt\simpset_{\Dim}(K_{i+1}\intsec\compcmplx{\ell})
=\simpset_{\Dim}(K_{i+1})$.
We also have that $K_i\subseteq K_{i+1}$
and $(K_{i+1}\intsec\compcmplx{\ell})\intersect K_i=K_{i}\intsec\compcmplx{\ell}$
for each $\ell$.
By Proposition~\ref{prop:dirsum-homo} and Remark~\ref{rmk:ext2gen},
we have the following commutative diagram:
\[
% \centerline
{\xymatrix{
\bigoplus_{\ell=1}^\compcnt\Hm_{\Dim-1}(K_i\intsec\compcmplx{\ell}) 
  \ar[d]^(.43){\rotatebox{90}{$\approx$}} \ar@{<->}[r]%^(.485){\psi_{i}} 
  & \bigoplus_{\ell=1}^\compcnt\Hm_{\Dim-1}(K_{i+1}\intsec\compcmplx{\ell}) 
  \ar[d]^(.43){\rotatebox{90}{$\approx$}} 
  \\ 
\Hm_{\Dim-1}(K_i)  \ar@{<->}[r]^(.5){\morph{\Fcal}{i}{\Dim-1}} 
  & \Hm_{\Dim-1}(K_{i+1}) 
  \\
}}
\]
% In the above diagram,
% $\psi_{i}$ is the direct sum of each linear map induced by the inclusion
% which connects $\Hm_{\Dim-1}(K_i\intsec\compcmplx{\ell})$. 
and therefore Proposition~\ref{prop:codim-1-mod-sum} follows.

\section{Conclusions}

In this paper, we propose near-linear algorithms for computing zigzag persistence
on graphs with the help of some dynamic graph data structures.
The algorithm for $0$-dimensional homology relates the computation to 
an algorithm in \cite{agarwal2006extreme} for
pairing critical points of Morse functions on $2$-manifolds, 
thereby giving a correctness proof for the algorithm in~\cite{agarwal2006extreme}.
%the pairing of critical points for Morse functions on $2$-manifolds~\cite{agarwal2006extreme};
The algorithm for $1$-dimensional homology reduces the computation to 
finding the earliest positive edge so that 
a $1$-cycle containing both the positive and negative edges 
resides in all intermediate graphs.
With the help of Alexander duality,
we also extend the algorithm for $0$-dimension
to compute the $(p-1)$-dimensional zigzag persistence for $\mathbb{R}^p$-embedded
complexes in near-linear time.

An obvious open question is
whether the time complexities can be improved.
Since the running time of our algorithms is determined in part 
by the data structures we use, 
one may naturally ask whether these data structure 
can be improved.

Furthermore, we note that the following may be helpful to future research efforts:
\begin{enumerate}
    \item In our algorithm for $0$-dimension, we build the barcode graph
    and utilize the algorithm in~\cite{agarwal2006extreme} to compute
    the zigzag barcode, which is also adopted by~\cite{dey2019computing}.
    Is there any other scenario in persistence computing
    where such a technique can be applied so that a more efficient algorithm
    can be derived?
    
    \item In our algorithm for $1$-dimension,
we utilize the dynamic-MSF
data structure~\cite{holm2001poly} for computing the max edge-weight
of the path connecting two vertices in an MSF.
The update of the
data structure which takes $O(\log^4 n)$ amortized time becomes the bottleneck of the 
algorithm.
However,
for the computation,
it can be verified that we only need to know the {\it minimax}\footnote{The {\it minimax} distance
of two vertices in a graph is the minimum of the max edge-weights
of all paths connecting the two vertices.} 
distance of two vertices in a graph.
Is there any faster way to compute the minimax distance in a dynamic graph?

    \item Another interesting question is whether our algorithms are more efficient
    practically when implemented %(currently we do not have an implementation),
    compared to some existing implementation~\cite{maria2014zigzag,Dionysus} 
    of zigzag persistence algorithms for general dimension.
\end{enumerate}

\bibliography{refs}

\end{document}